\documentclass[review,11pt]{elsarticle}
\usepackage{lineno}
\usepackage[stretch=20]{microtype}
\usepackage{amsmath, amssymb, amsthm, bm, booktabs, iftex, float, hyperref, tikz, quoting, xcolor}
\usetikzlibrary{automata, arrows.meta, calc, positioning, shapes.geometric}
\usepackage{newtx}
\usepackage{macros}
\allowdisplaybreaks
\journal{Artificial Intelligence}
\biboptions{sort&compress}

\begin{document}

\begin{frontmatter}

    \title{Cooperative Concurrent Games}
    \tnotetext[mytitlenote]{This paper is an extended and revised version of~\cite{Gutierrez2019}.}

    \author[monash]{Julian Gutierrez}
    \ead{julian.gutierrez@monash.edu}

    \author[oxford]{Szymon Kowara}

    \author[barilan]{Sarit Kraus}
    \ead{sarit@cs.biu.ac.il}
    
    \author[oxford]{Thomas Steeples}
    \ead{thomas.steeples@cs.ox.ac.uk}

    \author[oxford]{Michael~Wooldridge}
    \ead{michael.wooldridge@cs.ox.ac.uk}

    \cortext[corresponding]{Corresponding author: Julian Gutierrez.}
    \address[monash]{Faculty of Information Technology, Monash University, Australia}
    \address[barilan]{Department of Computer Science, Bar-Ilan University, Israel}
    \address[oxford]{Department of Computer Science, University of Oxford, United Kingdom}

    \begin{abstract}
  In \emph{rational verification}, the aim is to verify which temporal
  logic properties will obtain in a multi-agent system, under the
  assumption that agents (``players'') in the system 
  choose strategies for acting that form a game theoretic
  equilibrium. Preferences are typically defined by assuming
  that agents act in pursuit of individual goals, specified as
  temporal logic formulae. To date, rational verification has been
  studied using \emph{non-cooperative} solution concepts---Nash
  equilibrium and refinements thereof. Such non-cooperative solution
  concepts assume that there is no possibility of agents forming
  binding agreements to cooperate, and as such they are restricted in
  their applicability.  In this article, we extend rational verification
  to \emph{cooperative} solution concepts, as studied in the field of
  cooperative game theory.  We focus on the \emph{core}, as this is
  the most fundamental (and most widely studied) cooperative solution
  concept. We begin by presenting a variant of the core that seems
  well-suited to the concurrent game setting, and we show that this
  version of the core can be characterised using ATL\(^*\). We then
  study the computational complexity of key decision problems
  associated with the core, which range from problems in
  \textsc{PSpace} to problems in \textsc{3ExpTime}\@. We also
  investigate conditions that are sufficient to ensure that the core is 
  non-empty, and
  explore when it is invariant under bisimilarity. We then introduce
  and study a number of variants of the main definition of the core,
  leading to the issue of credible deviations, and to stronger notions
  of collective stable behaviour. Finally, we study cooperative rational 
  verification using an alternative model of preferences, in which players seek
  to maximise the mean-payoff they obtain over an infinite play in games where quantitative information is allowed.  
\end{abstract}

    \begin{keyword}
        Concurrent games; Cooperative games; Multi-agent
        Systems; Logic; Formal verification
    \end{keyword}
\end{frontmatter}

\section{Introduction}\label{secn:intro}

Intelligent agents such as Siri, Alexa, and Cortana are now used by millions of people every day. Originally proposed in the early 1990s, the vision for such agents was that they would be pro-active assistants, working in pursuit of goals delegated to them by their human users. Current intelligent agents make use of an ever-increasing array of AI technologies, starting with natural language understanding in their interface, used to communicate with human users. The next step in the development of such agents is expected to be the ability of such agents to interact not just with human users but \emph{with each other} in pursuit of their delegated goals. Such multi-agent systems raise a raft of research challenges, which have been the subject of considerable research over the past thirty years. The issue to which we address ourselves in this article is that of understanding the dynamics of such systems. In particular, collections of interacting agents may exhibit undesirable and unpredictable dynamics. The question of \emph{verifying} the possible behaviours of multi-agent systems thus naturally arises.

One increasingly popular approach to the problem of verifying multi-agent systems involves viewing a system as a game (in the sense of game theory), in which agents act rationally and strategically in pursuit of delegated goals. Given this, it is natural to ask what behaviours a system might exhibit under the assumption that agents act rationally---in accordance with game-theoretic solution concepts. This is the key idea that underpins the \emph{rational verification} paradigm~\cite{Gutierrez2015,Gutierrez2017a,WooldridgeGHMPT16,Abate2021}. 

Previous work on rational verification has used \emph{concurrent games} as a semantic model of multi-agent systems.
Since they were first introduced~\cite{Alur2002}, concurrent games have been very widely adopted in both the AI community and the verification/computer science community~\cite{Gutierrez2015,Gutierrez2017a,WooldridgeGHMPT16}. A concurrent game~\cite{Alur2002} is a finite-state environment, populated by a collection of independent, self-interested agents. A game takes place over an infinite sequence of rounds, where at each round, each agent chooses an action to perform. Preferences in concurrent games are typically modelled by assuming that each agent is associated with a temporal logic goal formula~\cite{Emerson1990}, which it desires to see satisfied.  The infinite plays generated by a game will either satisfy or fail to satisfy each player's goal, and players act rationally in an attempt to achieve their goal.  Since the satisfaction of a player's goal may be dependent on the choices made by other players, then players must make choices strategically in order to play optimally.

Now, in all previous studies that we are aware of, concurrent games are assumed to be \emph{non-cooperative}: players act alone, and binding agreements between players are ruled out. The game theoretic solution concepts used in previous studies of concurrent games have therefore been those studied in non-cooperative game theory---primarily Nash equilibrium and refinements thereof. In such a non-cooperative setting, the basic questions that we ask of a concurrent game are, for example, whether a particular temporal logic property holds on \emph{some} computation of the system that could arise through players selecting strategies that form a Nash equilibrium (the \textsc{E-Nash} problem) or whether a property holds on \emph{all} such computations (the \textsc{A-Nash} problem).  These problems can be understood as game-theoretic counterparts of the conventional model checking problem~\cite{Clarke2002}: in model checking, we are typically interested in whether a particular property could hold on some or all possible computations of a system, whereas in rational verification, we are interested in whether a property holds on some or all computations \emph{that could arise through rational choices on the part of the players}. The complexity of the corresponding decision problems in concurrent games has been extensively studied, and there now exist a small number of software tools that support rational verification~\cite{Gutierrez2018,Lomuscio2017,Brenguier2013,Kwiatkowska2020}.

Non-cooperative game theory, however, represents just one branch of game theory. \emph{Cooperative game theory} is a widely studied, albeit less well-known branch of game theory, which distinguishes itself from its non-cooperative counterpart in that it allows for the possibility that \emph{agents can make binding agreements with each other}. The possibility of binding agreements makes it possible for agents to work in teams, and to enjoy the benefits of cooperation. In the real world, we use contracts and other mechanisms (both formal and informal) to enable binding agreements. 

The aim of the present paper, therefore, is to extend the study of rational verification to include \emph{cooperative} solution concepts~\cite{OsborneR94,Maschler2013,Chalkiadakis2011}. Thus, we assume there is some (exogenous) mechanism through which players in a concurrent game can reach binding agreements and form teams (``coalitions'') in order to collectively achieve goals (although we emphasise that the nature of such a mechanism is beyond the scope of the present work). The possibility of binding cooperation and coalition formation eliminates some undesirable equilibria that arise in non-cooperative settings, and makes available a range of outcomes that cannot be achieved without cooperation. We focus on the \emph{core} as our key solution concept. The basic idea behind the core is that it consists of those strategy profiles from which no subset of players could benefit by collectively deviating. Now, in conventional cooperative games (characteristic function games with transferable utility~\cite{Chalkiadakis2011}), this intuition can be given a simple and natural formal definition, and as a consequence, the core is probably the most widely-studied solution concept for cooperative games. However, the conventional definition of the core does not naturally map into our concurrent game setting, because in our games, coalitions are subject to \emph{externalities}: whether a coalition has a beneficial deviation depends not just on the makeup of that coalition, \emph{but also on the behaviour of the remaining players}.

We begin by introducing the framework of concurrent games, and then proceed to define two variations of the core for such settings. In the first, a coalition of players is assumed to have a beneficial deviation if they have some course of action available to them which they would benefit from \emph{no matter what the remaining players did} (\emph{cf.}, the concept of the \(\alpha\)-\core\ in the classic game theory literature~\cite{uyanik2015nonemptiness}). However, this ``worst case'' (maximin) formulation of the core requires a deviation to be beneficial against \emph{all} courses of action by the remaining players---even those that the remaining agents would not rationally choose. This motivates a second definition, where a deviation is only required to be beneficial against all courses of action by remaining players that are \emph{credible}, in the sense that those players would then be no worse off than they were originally. We also consider games where the agents have \emph{quantitative} preferences, modelled as some kind of reward in every round of the game. In each case, we formally define the relevant solution concept, identify some of its key computational properties, give logical characterisations, and where possible, provide complexity results, which range from properties that can be checked in \textsc{PSpace} to properties that can be checked in \textsc{3ExpTime}\@. We also study model theoretic properties related to the \core: in particular, whether it is guaranteed to be non-empty, and whether temporal logic properties hold across bisimilar systems over plays (computation runs) induced by elements in the \core\ of the game (a highly desirable property from a formal verification perspective).

\noindent
\paragraph*{\textbf{Structure of the paper}} The remainder of this article is organised as follows:
\begin{itemize}
    \item In the following section, we summarise the key concepts from logic and concurrent games that are used throughout the paper.
    \item In Section~\ref{secn:core}, we define the \core\ and the main computational properties associated with it. In Section~\ref{secn:strong-core} we study the issue of credible coalition formations, with associated complexity results.
    \item In Section~\ref{secn:mean-payoff-games}, we study the core in the quantitative setting of concurrent mean-payoff games.
    \item Concluding remarks and related work are given in Section~\ref{secn:conc}, including a discussion around the implementation of our concepts using model checkers.
\end{itemize}

\section{Preliminaries}\label{secn:prelim}
\subsection{Set and Sequences}
Given any set \(S\), we use \(S^*\), \(S^\omega\), and \(S^+\) for, respectively, the sets of finite, infinite, and non-empty finite sequences of elements in \(S\). If \(w_1\in S^*\) and \(w_2\) is any other (finite or infinite) sequence, we write \(w_1w_2\) for their concatenation. For \(Q\subseteq S\), we write \(S_{-Q}\) for \(S\setminus Q\) and \(S_{-i}\) if \(Q=\set{i}\). We extend this notation to tuples \(u=(s_1,\ldots,s_k,\ldots,s_n)\) in \(S_1\times\cdots\times S_n\), and write \(u_{-k}\) for \((s_1,\ldots,s_{k-1},s_{k+1},\ldots,s_n)\), and similarly for sets of elements, that is, by~\(u_{-Q}\) we mean~\(u\) without each \(s_k\), for \(k\in Q\). Given a sequence \(w\), we write~\(w[t]\) for the element in position \(t+1\) in the sequence; for instance, \(w[0]\) is the first element of \(w\). We also use \emph{slice notation}: we write~\(w[l\ldots m]\) for the sequence \(w[l]\ldots w[m-1]\), \(w[l\ldots]\) for \(w[l]w[l+1]\ldots\), and \(w[\ldots m]\) for \(w[0]\ldots w[m-1]\); if \(m=0\), we let~\(w[l\ldots m]\) be the empty sequence, denoted~\(\epsilon\).

\subsection{Games} We begin by introducing the model of multi-agent systems that we use throughout the remainder of the paper: \emph{concurrent game structures}~\cite{Alur2002}. Informally, a concurrent game consist of a set of players, a set of actions for each of those players, a set of system states, and a transition function which describes how the state of the game changes, given a current state and an action for each of the players.

Formally, a concurrent game structure, \(M\), is given by a tuple,
\begin{equation*}
    M = (\Ag, \St, {\{\Ac_i\}}_{i \in \Ag}, s^0, \transf),
\end{equation*}
where:
\begin{itemize}
    \item \(\Ag\) and \(\St\) are finite, non-empty sets of \textbf{agents} and \textbf{states}, respectively---we usually identify \(\Ag\) with the set \(\{1, \ldots, n\}\);
    \item For each \(i\ \in \Ag\), \(\Ac_i\) is a finite, non-empty set of \textbf{actions} available to agent \(i\). We associate each state \(s\) with a set of actions available at that state, \(\Ac_i(s) \subseteq \Ac_i\), and we write \(\Ac\) for \(\Ac_1 \times \cdots \times \Ac_n\);
    \item \(s^0 \in \St\) is the \textbf{initial}/\textbf{start} state; and finally,
    \item \(\transf : \St \times \Ac \to \St\) is the \textbf{transition function} of the game.
\end{itemize}
As all the other components of the game can be derived from the transition function (the start state can be specified by listing its transitions first, say), the \emph{size} of \(M\) is the size of its transition function, given by \(|\St|\times|\Ac|^{|\Ag|}\).

Given a concurrent game structure \(M\), we can \emph{play} a game on it as follows: the game starts in state \(s^0\), and each player \(i \in \Ag\) chooses an action available to them, \(\ac_i^0 \in \Ac_i(s^0)\). The game then moves to a new state,
\begin{equation*}
    s^1 = \transf(s^0, \ac_1^0, \ldots, \ac_{n}^0).
\end{equation*}
This process is then repeated. We typically write \(s^i\) for the \(i^\text{th}\) state in the sequence, and \(\ac^i = (\ac_1^i, \ldots, \ac_n^i)\) for the \(i^\text{th}\) vector of actions played in the sequence. Thus for all \(t \in \Nat\), we have
\begin{equation*}
    s^{t+1} = \transf(s^t, \ac^t)
\end{equation*}

A \textbf{run}, \(\rho\) is a infinite sequence \(\run = s^0 s^1 s^2 \ldots\) such that for every \(t \in \Nat\), there exists some \(\ac \in \Ac\) such that \(s^{t+1} = \transf(s^t, \ac)\). A \textbf{path}, \(\pi\) is a finite prefix of a run.

A \textbf{strategy} for a player~\(i\) defines how player \(i\) chooses actions at each round of the game, as a function of the previous history of the game. Formally, a strategy is a function~\(\sigma_i\colon\St^+\to\Ac_i\) such that \(\sigma_i(\pi s)\in\Ac_{i}(s)\) for every \(\pi\in\St^*\) and \(s\in \St\). Thus, for every path \(\pi\), a strategy for a player \(i\) gives an action available to~\(i\) from the last state of that path. The set of strategies for player \(i\) is denoted by~\(\strategies_i\).\footnote{Here, we define strategies with respect to finite sequences of states. One can also define them in terms of finite sequences of action profiles. As the sequence of states can be derived from the action profiles, but not the other way around, these strategies are more powerful than the ones we use, and indeed, provide desirable properties such as the invariance of Nash equilibria under bisimilarity~\cite{Gutierrez2017b}. However, strategies defined relative to states are standard in the concurrent game structures~\cite{Alur2002} and rational verification~\cite{Gutierrez2017a} literature, so we use this model for consistency. For more details and further discussion, refer to~\cite{Gutierrez2017b}.}

A \textbf{strategy profile} \(\vec{\sigma}\) is a tuple of strategies, one for each player: \(\vec\sigma=(\sigma_1,\dots,\sigma_n) \in \strategies_1\times\cdots\times\strategies_n\). Observe that a strategy profile, \(\vec\sigma\), together with a state \(s\), induces a unique run, \(\run(\vec\sigma,s)\), where:
\begin{itemize}
    \item \(\run(\vec\sigma,s)[0]=s\); and
    \item \(\run(\vec\sigma,s)[t+1]=\transf(\run(\vec\sigma,s)[t],\sigma_1(\run(\vec\sigma,s)[\ldots t]),\ldots,\sigma_n(\run(\vec\sigma,s)[\ldots t]))\), for all \(t\in \Nat\).
\end{itemize}
We write \(\run(\vec\sigma)\) if \(s=s^0\).

Note that viewing strategies as functions \(\sigma_i\colon\St^+\to\Ac_i\) is problematic with respect to computational analysis, because the domain of such a function is infinite. To be able to answer questions relating to (for example) computational complexity, we need a finite representation for strategies that must operate over an infinite number of rounds. For this purpose, we follow standard practice in the concurrent game's literature, and model strategies as \textbf{finite state machines with output} (transducers)~\cite{Gutierrez2015,Gutierrez2017a}. We note that, for players with LTL goals, such strategies are sufficient: no more powerful model of strategies is necessary~\cite{Gutierrez2015,Gutierrez2017a}. Formally, a strategy for player~\(i\) is a structure \[\sigma_i = (Q_i, q_i^0, \strtransf_i, \stroutf_i)\] where:
\begin{itemize}
    \item \(Q_i\) is a finite, non-empty set of \textbf{strategy states};
    \item \(q^0_i \in Q_i\) is the \textbf{initial} strategy state;
    \item \(\strtransf_i : Q_i\times \St \to Q_i\) is a \textbf{transition function}; and
    \item \(\stroutf_i : Q_i \to \Ac_i\) is an \textbf{output function}.
\end{itemize}

A machine strategy works as follows: it begins in the initial state, \(q_i^0\), and chooses an action based on this, \(\stroutf(q_i^0)\). The state of the game then follows the transition function into a new state. Based on this new (game) state, and the state of the strategy, the strategy then also moves into a new state, based on the strategy transition function. This process repeats, yielding a new action at each timestep. It is easy to see that a strategy profile consisting solely of machine strategies will be eventually periodic (i.e., it will eventually enter a configuration that it was in previously, at which point it will start to repeat its behaviour).

Sometimes, we find we can use an even simpler model of strategies: \textbf{memoryless strategies}. A memoryless strategy \(\sigma_i\colon\St\to\Ac_i\) simply chooses an action based on the current state of the environment. Whilst memoryless strategies are not nearly as expressive as finite-memory strategies, they are still of great importance, owing to their conceptual simplicity, their ease of implementation, and the fact that in many types of game (such as two-player mean-payoff games~\cite{Ehrenfeucht1979}, and parity games~\cite{Emerson1991}), memoryless strategies are sufficient to act optimally.\footnote{For a characterisation of the types of games where memoryless strategies are sufficient for one or both players to play optimally, refer to~\cite{Gimbert2005,Kopczynski2006}.}

\subsection{Logics}

Broadly, we need to appeal to certain logics for two main reasons: expressing the preferences of agents, and reasoning about those agents and their preferences. For this, we will use three logics throughout this paper: \emph{Linear Temporal Logic} (LTL)~\cite{Pnueli77}, \emph{Alternating-Time Temporal Logic} (\atlstar)~\cite{Alur2002}, and \emph{Strategy Logic} (SL)~\cite{MogaveroMPV14}. We use LTL for modelling agent preferences, and \atlstar and SL for forming logical characterisations of the game-theoretic concepts we will study. Whilst it can be seen more formally from the syntax and semantics presented below, we note here that LTL can be seen as a subset of \atlstar, which can be seen as a subset of SL.

\subsubsection{Linear Temporal Logic}
LTL is a widely used logic for reasoning about the behaviours of concurrent systems, and while we present the key concepts here, we refer the reader to any standard temporal logic textbook for details (\emph{e.g.},~\cite{Baier2008}).

Let \(\AP\) be a set of propositional variables. Then the syntax of an LTL formula \(\varphi\) is given by the following grammar:
\begin{equation*}
    \varphi := p \mid \lnot \varphi \mid \varphi \lor \varphi \mid \temporalNext \varphi \mid \varphi \temporalUntil \varphi,
\end{equation*}
where \(p \in \AP\). The set of all LTL formulae over a set of propositional variables \(\AP\) is denoted \(\mathcal{L}(\AP)\). If \(\AP\) is clear from the context, we may instead just write \(\mathcal{L}\). We also introduce the traditional propositional abbreviations, \(\cdot \land \cdot\), \(\cdot \rightarrow \cdot\), \(\cdot \leftrightarrow \cdot\), defined in the usual way, as well as the abbreviations \(\temporalEventually \varphi\) for \(\true \temporalUntil \varphi\), and \(\temporalAlways \varphi\) for \(\lnot \temporalEventually \lnot \varphi\).

Typically, in the context of model checking, the semantics of LTL formulae are defined relative to labelled transition systems~\cite{Keller1976}, or Kripke structures~\cite{Kripke1963}, but for our purposes, we define them with respect to the infinite runs generated by concurrent game structures. Formally, let \(M\) be a concurrent game structure, and let \(\lambda : \St \to \powerset{\AP}\) be a labelling function, mapping states to sets of propositional variables. Then given an infinite run \(\rho \in \St^\omega\) and an LTL formula \(\varphi\), we say that \((M, \lambda, \rho)\) \emph{models} \(\varphi\) and write \((M, \lambda, \rho) \models \varphi\) according to the following inductive definition:
\begin{itemize}
    \item For \(p \in \AP\), we have \((M, \lambda, \rho) \models p\) if and only if \(p \in \lambda(\rho[0])\);
    \item For \(\varphi \in \mathcal{L}(\AP)\), we have \((M, \lambda, \rho) \models \lnot \varphi\) if and only if it is not the case that \((M, \lambda, \rho) \models \varphi\);
    \item For \(\varphi, \psi \in \mathcal{L}(\AP)\), we have \((M, \lambda, \rho) \models \varphi \land \psi\) if and only if we have both \((M, \lambda, \rho) \models \varphi\) and \((M, \lambda, \rho) \models \psi\);
    \item For \(\varphi \in \mathcal{L}(\AP)\), we have \((M, \lambda, \rho) \models \temporalNext \varphi\) if and only if  \((M, \lambda, \rho[1\ldots]) \models \varphi\);
    \item For \(\varphi, \psi \in \mathcal{L}(\AP)\), we have \((M, \lambda, \rho) \models \varphi \temporalUntil \psi\) if and only if there exists some \(j \geq 0\) such that \((M, \lambda, \rho[j\ldots]) \models \psi\) and for all \(i < j\) we have \((M, \lambda, \rho[i\ldots]) \models \varphi\).
\end{itemize}
For notational convenience, we will write \((G, \run) \models \varphi\) as shorthand for \((M, \lambda, \run) \models \varphi\), and if \(G\) is apparent from the context, we will just write \(\run \models \varphi\).

\subsubsection{Alternating-Time Temporal Logic}
\atlstar is an extension of the branching-time temporal logic CTL\textsuperscript{*}~\cite{Emerson1986} --- we shall use this logic in our proofs to reason about coalitions of agents more effectively. The key operator in \atlstar is the cooperation modality \(\ATLEpath{C}\phi\), which asserts that the coalition \(C\) has the power to enforce the temporal property \(\phi\); more specifically, that there is a collection of strategies for \(C\) such that if \(C\) follow these strategies, then no matter what other agents do, \(\phi\) is guaranteed to be made true. Given a set of atomic propositions~\(\AP\) and a set of agents~\(\Ag\), the language of ATL\textsuperscript{*} formulae is defined by the following grammar:
\begin{align*}
    \varphi  ::= & p  \mid  \lnot \varphi \mid \varphi \lor \varphi \mid  \ATLEpath{C} \psi                         \\
    \psi ::=     & \varphi \mid \lnot \psi \mid \psi \lor \psi \mid \temporalNext \psi \mid \psi\temporalUntil \psi
\end{align*}
where \(p\in\AP\) and \(C\subseteq\Ag\). We call the formulae produced by \(\varphi\) in the above grammar \textbf{\atlstar state formulae}, and denote the set that contains them by \(\mathcal{L}_s(\AP, \Ag)\) and those generated by \(\psi\) in the above grammar \textbf{\atlstar path formulae}, denoted by \(\mathcal{L}_p(\AP, \Ag)\). (They are given these names as their semantics are defined relative to states and paths respectively.) We emphasise that only \atlstar state formulae are well-formed \atlstar formulae, and thus, we write \(\mathcal{L}(\AP, \Ag)\) as shorthand for \(\mathcal{L}_s(\AP, \Ag)\). When either \(\AP\) or \(\Ag\), or both, are known, we may omit them. With \(\AP'\subseteq\AP\), we may write \(\varphi|_{\AP'}\) if \(\varphi\in\loglang(\AP',\Ag)\) for some set of agents \(\Ag\).

In addition to the abbreviations used for LTL formulae as described above, we also use the shorthands \(\Epath\varphi\) for \(\ATLEpath{\Ag}\varphi\), \(\Apath\varphi\) for \(\ATLEpath\emptyset\varphi\), and \(\ATLApath{C}\varphi\) for \(\neg\ATLEpath{C}\neg\varphi\). Finally, we define the \emph{size} of an \atlstar formulae \(\varphi\) as its number of subformulae.

To define the semantics of \atlstar formulae, we actually need to define two semantic relations, \(\models_s\) (for state formulae) and \(\models_p\) (for path formulae). So, let \(M\) be some concurrent game structure, along with a labelling function \(\lambda : \St \to \powerset{\AP}\). Then given a state, \(s \in St\) and an \atlstar formula \(\varphi\), we say that \((M, \lambda, s)\) \emph{models} \(\varphi\) and write \((M, \lambda, s) \models \varphi\) according to the following inductive definition:
\begin{itemize}
\item For \(\varphi \in \mathcal{L}_s(\AP, \Ag)\), we have \((M, \lambda, s) \models \varphi\) if and only if \((M, \lambda, s) \models_s \varphi\);
  \item For operators that lie in LTL, their inductive semantics are the same as in LTL;
    \item For \(\varphi \in \mathcal{L}_p(\AP, \Ag)\), we have \((M, \lambda, s) \models_s \ATLEpath{C} \varphi\) if and only if there is some strategy vector \(\vec{\sigma}_C\) for the coalition \(C\), such that for all complementary strategy profiles, \(\vec{\sigma}_{\Ag \setminus C}\), it is the case that \((M, \lambda, \rho((\vec{\sigma}_{\Ag \setminus C}, \vec{\sigma}_C), s)) \models_p \varphi\) holds;
    \item For \(\varphi \in \mathcal{L}_s(\AP, \Ag)\), we have \((M, \lambda, \rho) \models_p \varphi\) if and only if \((M, \lambda, \rho[0]) \models_s \varphi\);
\end{itemize}
Where the concurrent game and labelling function are clear from context, we will simply write \(s\models\phi\). Given a concurrent game structure \(M\) and a labelling function \(\lambda : \St \to \powerset{\AP}\), we say that \(\varphi\) is \textbf{satisfiable} if there exists some state \(s \in \St\) such that \((M, \lambda, s) \models \varphi\). Moreover, we say that \(\varphi\) is \textbf{equivalent} to \(\psi\) if for all states \(s \in \St\) we have \((M, \lambda, s) \models \varphi\) if and only if \((M, \lambda, s) \models \psi\).

Note that LTL can been seen as the sublogic of \atlstar given by all formulae \(\Apath \varphi\), where the formula \(\varphi\) does not contain the ``coalition'' quantifiers \(\ATLEpath{C}\) or \(\ATLApath{C}\). Thus \atlstar is a particularly effective tool for reasoning about the LTL properties that coalitions can achieve, and this is exactly how we will use it when we come to prove the complexity bounds of our decision problems. Specifically, if we have an LTL game \(G\), we can write \((G, s) \models \varphi\) as shorthand for \((M, \lambda, s) \models \varphi\), and furthermore, if the game \(G\) is apparent from the context, we shall drop it and simply write \(s \models \varphi\) instead.

The relevant decision problem here is the \textbf{model checking problem for \atlstar}, which we utilise heavily in the following section:

\begin{quote}
    \underline{\textsc{\atlstar Model Checking}}:\\
    \emph{Given}: Concurrent game \(M\), labelling function \(\lambda\), state \(s \in \St\), and \atlstar{} formula \(\varphi\).\\
    \emph{Question}: Is it the case that \((M,\lambda,s) \models \varphi\)?
\end{quote}

This problem is \textsc{2ExpTime}-complete~\cite{Alur2002} for games with two or more players, and \textsc{PSpace}-complete for one-player games (as then, the problem reduces to LTL model checking)~\cite{SistlaC85}.

\subsubsection{Strategy Logic}
Whilst \atlstar\ is a powerful tool for reasoning about coalitions, it famously cannot represent certain game-theoretic concepts, such as the Nash equilibrium~\cite{MogaveroMPV14} --- explicit quantification over strategies is required to do this. For our purposes, we will be able to make most of the necessary logical characterisations using \atlstar, but when we come to the \emph{strong core} in Section~\ref{secn:strong-core}, we too will also need the ability to reason explicitly about strategies. SL extends LTL with two \textbf{strategy quantifiers}, \(\EExs{x}\) and \(\AAll{x}\), and an \textbf{agent binding} operator \((i, x)\), where \(i\) is an agent and \(x\) is a variable. These operators can be read as ``\emph{there exists a strategy \(x\)}'', ``\emph{for every strategy \(x\)}'', and ``\emph{bind agent \(i\) to the strategy associated with variable \(x\)}'', respectively. Formally, SL formulae are inductively built from a set of propositions \(\AP\), variables \(\VarSet\), and agents \(\Ag\), using the following grammar, where \(p \in \AP\), \(x \in \VarSet\), and \(i \in \Ag\):
\begin{equation*}
    \varphi ::=
    p \mid \neg \varphi \mid \varphi \wedge \varphi \mid \temporalNext \varphi \mid \varphi \temporalUntil \varphi \mid \EExs{x} \varphi \mid \AAll{x} \varphi \mid (i, x) \varphi.
  \end{equation*}
  
We can now present the semantics of SL\@. First, denoting the set of all strategies by \(\StrSet\), we define an assignment to be a partial function that maps variables and agents to strategies, \(\asgFun \in \AsgSet = (\VarSet \cup \Ag) \rightharpoonup \StrSet\). We then use the notation \(\asgFun[i \mapsto f] / \asgFun[x \mapsto f]\) to refer to the assignment which equals \(f\) on \(i/x\), and agrees with \(\asgFun\) on every other input on which it is defined. Then, given a concurrent game structure~\(M\), for all SL formulae \(\varphi\), states \(s \in \St\) in \(M\), and assignments \(\asgFun \in \AsgSet\), the relation \(M, \asgFun, s \models \varphi\) is defined as follows: 
\begin{enumerate}
    \item For LTL formulae embedded in SL, their inductive semantics are the same as in LTL\@;
    \item\label{def:sl(semantics:qnt)} For all formulae \(\varphi\) and variables \(x \in \VarSet\) we have:
          \begin{enumerate}
              \item\label{def:sl(semantics:eqnt)}
              \(M, \asgFun, s \models \EExs{x} \varphi\) if and only if there exists some \(\strFun \in \StrSet\) such that \(M,\asgFun[x \mapsto \strFun], s \models \varphi\);
              \item\label{def:sl(semantics:aqnt)}
              \(M, \asgFun, s \models \AAll{x} \varphi\) if and only if for all \(\strFun \in \StrSet\) we have \(M,\asgFun[x \mapsto \strFun], s \models \varphi\);
          \end{enumerate}
          \item\label{def:sl(semantics:bnd)} For every agent \(i \in \Ag\) and variable \(x \in \VarSet\), we have \(\ M, \asgFun, s \models (i, x) \varphi \ \) if and only if\(\ M, \asgFun[i \mapsto \asgFun(x)], s \models \varphi \ \).
\end{enumerate}
For a sentence \(\varphi\), that is, a formula with no free variables and agents~\cite{MogaveroMPV14}, we say that \(M\) satisfies \(\varphi\), and write \(M \models \varphi\) in that case, if \(M, \emptyset, s^0 \models \varphi\), where \(\emptyset\) is the empty assignment. We use the following abbreviations: \(\SLE{i}\varphi\) for \(\EExs{x}(i,x)\varphi\) and \(\SLA{i}\varphi\) for \(\AAll{x}(i,x)\varphi\), which can be intuitively understood as ``there is a strategy for agent~\(i\) such that \(\varphi\) holds'' and ``\(\varphi\) holds, for all strategies of agent~\(i\)'', respectively. We extend this notation to sets of players and write, for instance, \(\SLE{C}\varphi\) instead of \(\SLE{i}\ldots\SLE{j}\varphi\), where \(C = \set{i,\ldots,j}\), and similarly for the universal quantifier operator. Then, with \(\SLE{C}\varphi\) we mean that \emph{``coalition~\(C\) has a joint strategy such that \(\varphi\) holds.''}

\subsection{LTL Games} We can now define \emph{LTL games}~\cite{Pnueli1989,Fisman2010,Gutierrez2015}. The key idea in an LTL game is that each player \(i\) is associated with an LTL goal formula \(\gamma_i\), which it desires to see satisfied. Formally, an LTL game, \(G\), is given by a structure
\begin{equation*}
    G=(M, \AP, \lambda, {(\gamma_i)}_{i \in \Ag}),
\end{equation*}
where \(M\) is a concurrent game structure, \(\AP\) is a set of atomic propositions, \(\lambda : \St \to \powerset{\AP}\) is a labelling function, and for each \(i \in \Ag\), \(\gamma_i\) is an LTL formula over \(\AP\) that defines that player's preference relation over runs.

We now describe how temporal goal formulae \(\gamma_i\) induce \textbf{preference relations} \(\succeq_i\) over runs. First, given that under a provided run, LTL goals are either satisfied or not, we can identify a set of ``winners'' and a set of ``losers'' for that run. Formally, let \(\Winners(\run)\) denote the set of players that get their goal achieved under \(\run\), and let \(\Losers(\run)\) denote the set of players that do not:
\begin{align*}
    \Winners(\run) & = \{i \in \Ag \mid \run \models \gamma_i\} \\
    \Losers(\run)  & = \Ag \setminus \Winners(\run).
\end{align*}

We can now use winners and losers to define the preference relations of our agents. Intuitively, agent \(i\) will always strictly prefer a run that satisfies its goal \(\gamma_i\) over one that does not, but is indifferent between two runs that satisfy its goal, and is indifferent between two runs that fail to satisfy its goal. Formally, for two runs \(\run, \run'\), and a player \(i\), we have \(\run \succeq_i \run'\) if and only if \(i\) is a winner under \(\run\) (note the lack of dependence of \(\run^\prime\)) or \(i\) is a loser under both \(\run\) and \(\run^\prime\). Put alternatively, we have \(\run \succeq_i \run'\) if and only if \((M, \lambda, \run') \models \gamma_i\) implies that \((M, \lambda, \run) \models \gamma_i\). Strict preference relations \(\succ_i\) are defined in the standard way: \(\run \succ_i \run'\) if and only if \(\run \succeq_i \run'\) but not \(\run' \succeq_i \run\). We leave the reader to verify, first, that this definition matches our informal explanation of preferences above, and second, that the relations \(\succeq_i\) so defined are indeed preference relations (\emph{i.e.}, the binary relation \(\succeq_i\) is reflexive, complete, and transitive). Finally, we note that we use \(\mathcal{L}\) in two different senses in this paper: to denote the languages of various logics, and to denote the set of losers of a run. As these generally appear in different contexts, and take different inputs, there should be no ambiguity as to what interpretation should be taken.

With preference relations defined, we can introduce game-theoretic solution concepts. For now, we will stick with non-cooperative solution concepts. First, a strategy profile \(\vec{\sigma}\) is said to be a \textbf{Nash equilibrium} if there is no player \(i \in \Ag\) and strategy \(\sigma_i'\) for \(i\) such that we have \((\vec{\sigma}_{-i}, \sigma_i') \succ_i \vec{\sigma}\). That is, \(\vec{\sigma}\) is a Nash equilibrium if no player can benefit by unilaterally changing its strategy (assuming all other players leave their strategies unchanged). Let \(\NE(G)\) denote the set of Nash equilibria of the game \(G\).

We emphasise that Nash equilibrium only considers \emph{unilateral} deviations, \emph{i.e.}, deviations by individual players. Compare this to the notion of a \textbf{strong Nash equilibrium}~\cite{Aumann1959,Aumann1960}: a strategy profile \(\vec{\sigma}\) is a strong Nash equilibrium if there is no coalition \(C\) and no strategy \(\sigma_C'\) such that for all \(i \in C\) we have \((\vec{\sigma}_{-i}, \sigma_i') \succ_i \vec{\sigma}\). Thus, strong Nash equilibria are those strategy profiles that are immune to multilateral deviations.

Finally, we mention in passing the notion of the \emph{coalition-proof Nash equilibrium}~\cite{Bernheim1987a,Bernheim1987b}. Whilst we refrain from giving a formal definition here, informally, this can be thought of as capturing the outcomes that would arise when players have unlimited, but non-binding, pre-play communication. 

\section{Cooperative Rational Verification}\label{secn:core}

\subsection{Defining the Core}

We want to define counterparts of the rational verification problems \textsc{E-Nash} and \textsc{A-Nash}, as studied in~\cite{Gutierrez2015,Gutierrez2017a}, but for cooperative settings. For this, we need a version of the \core\ for our concurrent game setting. The \core\ is probably the best-known solution concept in cooperative game theory. Like Nash equilibrium in the non-cooperative setting, the \core\ defines a notion of stability for games, but whereas Nash equilibrium only requires that an outcome is stable in the sense that it admits no \emph{individual} beneficial deviations, the \core\ requires that an outcome admits no beneficial deviations by \emph{coalitions}. In the ``standard'' model of cooperative games, (cooperative games with transferable utility in characteristic function form~\cite{Chalkiadakis2011}), this intuition is easily formalised, but in concurrent games, there is an important difficulty, as follows.

Suppose a coalition of players \(C\subseteq \Ag\) are contemplating participating in a strategy profile \(\vec{\sigma}\), and in particular, are attempting to determine whether they have a cooperative beneficial deviation from \(\vec{\sigma}\). Now, as they consider possible beneficial deviations --- collective strategies \(\vec{\sigma}_{C} \in \prod_{i \in C} \Sigma_i\) --- \emph{what assumptions should \(C\) make about the behaviour of the remaining players \(\Ag \setminus C\)?} In particular, assuming that the remaining players will not respond, by potentially altering their strategy, is implausible in a cooperative setting.\footnote{This is the kind of behaviour that one has to assume to define strong Nash equilibrium, a non-cooperative solution concept.} Rational players who can cooperate will respond to the deviation rationally and in a cooperative way against the players in \(C\). And, crucially, whether or not \(C\)'s putative deviation is in fact beneficial may well depend upon the behaviour of the remaining players. In game theoretic terms, our concurrent game setting is subject to \emph{externalities}: the performance of the coalition \(C\) depends not just on the makeup of the coalition \(C\), but on the behaviour of the remaining players.

It is well-known that cooperative solution concepts are difficult to define in the presence of externalities~\cite{Chalkiadakis2011}. In particular, there is no universally accepted definition of the \core\ for games with externalities. Our first definition of the \core\ for concurrent games, therefore, captures \emph{worst case} (maximin) reasoning. Thus, when coalition \(C\) is contemplating a deviation, it requires that this deviation will be beneficial \emph{no matter what the remaining players do}. This idea has been explored in the concept of the \(\alpha\)-\core\ in cooperative games~\cite{uyanik2015nonemptiness}. To make it formal, we need to formalise the notion of a beneficial deviation. Let \(\vec{\sigma}\) be a strategy profile and let \(C\) be a coalition; then we say that \(\vec{\sigma}_{C}'\) is a beneficial deviation from \(\vec{\sigma}\) if both:
\begin{enumerate}
    \item \(C \subseteq \Losers(\rho(\vec{\sigma}))\); and
    \item For all \(\vec{\sigma'}_{-C}\), we have \(C \subseteq \Winners(\rho(\vec{\sigma}_{C}',\vec{\sigma}_{-C}'))\).
\end{enumerate}
In other words, \(\vec{\sigma}_{C}'\) is said to be a beneficial deviation from \(\vec{\sigma}\) if the players in \(C\) would be better off playing \(\vec{\sigma}_{C}'\), rather than their respective strategies in \(\vec{\sigma}\), \emph{no matter what strategies the players outside \(C\) play}. The \core\ of a game \(G\), denoted \(\coreset{G}\), is then defined to be the set of outcomes of \(G\) that no coalition has a beneficial deviation from. 

\begin{example}\label{ex:poorne}
    Consider the following game, which illustrates how cooperative and non-cooperative solution concepts differ: it contains a poor quality Nash equilibrium that is not in the \core\@. The ability to cooperate makes it possible for agents to avoid the undesirable equilibrium. The game contains two players, \(\Ag = \set{1,2}\) and two variables \(\AP = \set{p,q}\), with player~\(1\)'s action set being \(\Ac_1 = \set{pt,pf}\) and player~2's action set being \(\Ac_2 = \set{qt,qf}\), satisfying that, for every reachable state, if player~1/2 plays \(pt/qt\) then \(p/q\) will hold, and will not hold if \(pf/qf\) is played instead (\emph{i.e.}, player~1 ``controls'' the value of \(p\) and player~\(2\) the value of \(q\)). Their goals are identical (and so the game is a coordination game): \(\gamma_1 = \gamma_2 = \temporalAlways(p \wedge q)\). Now, consider the strategy profile \(\vec{\sigma}\) in which both players simply fix their respective variables to be false forever (\emph{i.e.}, play \(pf\) and \(qf\) forever). Neither player will have their goal achieved by such a strategy profile. However, the strategy profile forms a Nash equilibrium, because unilateral deviation cannot improve the situation: neither player has an alternative strategy which would make them better off. In fact, there are infinitely many such poor-quality Nash equilibria in this game, where neither player gets their goal achieved. However, this strategy profile is \emph{not} in the \core, because there is a cooperative beneficial deviation to the strategy profile in which both players fix their variables to be true forever (\emph{i.e.}, play \(pt\) and \(qt\) forever). And, in fact, in every strategy profile which lies in the core, both players get their goal achieved. Thus, using the core instead of Nash equilibrium eliminates poor quality equilibria from the game, leading to socially more desirable outcomes.
\end{example}

Before proceeding, we note that additional variations of the core have been proposed in the game theory literature on cooperative games with externalities, in particular, the \(\beta\)-core and \(\gamma\)-core. The \(\beta\)-core assumes that players outside a deviating coalition will choose strategies to maximally punish the deviating coalition, and that members of a deviating coalition will then collectively best respond to these punishment strategies. Thus, the \(\beta\)-core assumes a \emph{minimax} choice on behalf of the deviating coalition. In contrast, the \(\gamma\)-core assumes that players outside a deviating coalition will form singleton coalitions, with each non-deviator choosing individual strategies in a best response to the deviating coalition strategies. We note that these solution concepts are much less studied than the \(\alpha\)-core, and are arguably less natural and less well-motivated. We, therefore, leave them for future study (see also comments in the Related Work Section~\ref{secn:conc}).

\subsection{Decision Problems}

In rational verification~\cite{Gutierrez2015,Gutierrez2017a,WooldridgeGHMPT16} we are mainly interested in checking which temporal logic properties are satisfied in a given solution concept of a game; typically, in the non-cooperative setting, we study what LTL formulae hold in the Nash equilibria \(\NE(G)\) of a game \(G\). In the cooperative setting, as introduced here, we are instead interested in what properties hold in the \core\ of the game. The two main decision problems in rational verification are checking whether a temporal logic formula is satisfied by some/every stable strategy profile of the game. For the core, these problems are defined as follows---\emph{cf.}~\cite{Gutierrez2015,WooldridgeGHMPT16,Gutierrez2017a}.
\begin{quote}
    \underline{\ecore}:\\
    \emph{Given}: Game \(G\), LTL formula \(\varphi\).\\
    \emph{Question}: Is it the case that \(\exists \vec{\sigma} \in \coreset{G} : \run(\vec{\sigma})\models\varphi\)?
\end{quote}

\begin{quote}
    \underline{\acore}:\\
    \emph{Given}: Game \(G\), LTL formula \(\varphi\).\\
    \emph{Question}: Is it the case that  \(\forall \vec{\sigma} \in \coreset{G} : \run(\vec{\sigma})\models\varphi\)?
\end{quote}

One decision problem considered in the non-cooperative setting is the \textsc{Non-Emptiness} problem, which asks if the set of Nash equilibria of a given game is non-empty. However, as we shall show momentarily, the core of the game is always non-empty, so it does not make sense to consider the corresponding problem in the cooperative setting.

We will also be interested in two additional decision problems: checking whether a given strategy profile is in the \core\ (\cmembership), and checking whether a given strategy vector for a coalition is a beneficial deviation with respect to a strategy profile (\bdeviation). Formally:

\begin{quote}
    \underline{\cmembership}:\\
    \emph{Given}: Game \(G\), strategy profile \(\vec{\sigma}\).\\
    \emph{Question}: Is it the case that \(\vec{\sigma} \in \coreset{G}\)?
\end{quote}

\begin{quote}
    \underline{\bdeviation}:\\
    \emph{Given}: Game \(G\), strategy profile \(\vec{\sigma}\), coalition \(C\), and deviation \(\vec{\sigma}_{C}'\).\\
    \emph{Question}: Is \(\vec{\sigma}_{C}'\) a beneficial deviation from \(\vec{\sigma}\)?
\end{quote}

In what follows it is helpful to use the concept of a \textbf{fulfilled coalition}: a coalition is fulfilled if they are able to achieve their goals \emph{irrespective of what other players do}; that is, they have a collective strategy that will guarantee their goals are achieved. Formally, we say that a coalition of players~\(C\) is fulfilled if there is a joint strategy~\(\vec{\sigma}_C\) for~\(C\subseteq \Ag\) such that for all joint strategies \(\vec{\sigma}_{-C}\) for~\(\Ag\setminus C\) we have
\begin{equation*}
  \run((\vec{\sigma}_C,\vec{\sigma}_{-C}))\models\bigwedge_{i\in C}\gamma_i.
\end{equation*}
Thus, a fulfilled coalition has a winning strategy to collectively achieve their goals. Since we are considering cooperative games, the question is whether such a coalition will form. We have the following:
\begin{lemma}\label{lem:fc}
    \begin{enumerate}
        \item There are games \(G\), with strategy profiles \(\vec{\sigma}\in\coreset{G}\), containing fulfilled coalitions \(C\subseteq\Ag\) such that \(C\not\subseteq\Winners(\rho(\vec{\sigma}))\);
        \item For every game \(G\), strategy profile \(\vec{\sigma}\in\coreset{G}\), and fulfilled coalition \(C\), we have that \(C\cap\Winners(\rho(\vec{\sigma}))\neq\emptyset\);
        \item For every game~\(G\) and fulfilled coalition~\(C\), then there is~\(\vec{\sigma}\in\coreset{G}\) such that \(C\subseteq\Winners(\rho(\vec{\sigma}))\).
    \end{enumerate}
\end{lemma}
Informally, the first part of the lemma says that the fact that a coalition is fulfilled does not mean that every player in such a coalition is guaranteed to get its goal achieved under an arbitrary member of the core. However, the second part of the lemma says that in any member of the core, some agents of every fulfilled coalition must get their goals achieved. And, the third part of the lemma says that for every fulfilled coalition the \core\ contains a strategy profile in which every player of this coalition gets its goal achieved.
\begin{proof}[Proof of~\ref{lem:fc}.1]
  Consider a game with three states, \(\start, \up, \down\), and three players, \(\{1,2,3\}\). Let \(\start\) be the start state of the game,  and assume that players \(2\) and \(3\) have only one action available to them, and player \(1\) has two actions, \(\HEADS\) and \(\TAILS\). As such, the transition function can be defined solely in terms of player \(1\)'s action, and we let \(\transf(\start, \HEADS) = \up\) and \(\transf(\start, \TAILS) = \down\). The remaining two states are sink states, with \(\transf(\up, \HEADS) = \transf(\up, \TAILS) = \up\) and \(\transf(\down, \HEADS) = \transf(\down, \TAILS) = \down\). Further suppose we have \(\AP = \{p,q\} \), and that \(\lambda(\start) = \emptyset\), \(\lambda(\up) = \{p\}\) and \(\lambda(\down) = \{q\} \). Finally, let the goals of the players be as follows: \(\gamma_1 = \true\), \(\gamma_2 = \ltlX \ltlG p\), and \(\gamma_3 = \ltlX \ltlG q\). As the game will always end in one of two sink states after the first action, we can identify strategies simply by the first action of player 1: \(\HEADS\) or \(\TAILS\).

  Now, note that the coalition \(\{1,3\}\) is a fulfilled coalition: if player 1 plays \(\TAILS\), then the game will end up in the state \(\down\), and will permanently remain there, satisfying player 1's and player 3's goals. However, note that the strategy profile where player 1 plays \(\HEADS\) is a member of the core: player 1's and player 2's goals are satisfied, and as player 3 only has one action, they cannot deviate to get their goal achieved. Thus, we have a fulfilled coalition and a member of the core, such that the fulfilled coalition is not a subset of the winners of the member of the core. Note the same argument could have been applied the other way around --- \(\{1,2\}\) is a fulfilled coalition, and \(\TAILS\) is a member of the core whose winners are not a superset of \(\{1,2\}\).
\end{proof}

\begin{proof}[Proof of~\ref{lem:fc}.2]
      Let \(G\) be a game, \(\vec{\sigma}\) be a strategy profile that lies in the core, and \(C\) a fulfilled coalition. For the sake of a contradiction, further suppose that the intersection of \(C\) and \(\Winners(\rho(\vec{\sigma}))\) is empty. This means that for each \(i \in C\), we have \(\rho(\vec{\sigma}) \not \models \gamma_i\). But as \(C\) is a fulfilled coalition, there exists a strategy \(\vec{\sigma}_C^\prime\), such that for all counterstrategies \(\vec{\sigma}_{-C}\), we have \(\run((\vec{\sigma}_C,\vec{\sigma}_{-C}))\models\bigwedge_{i\in C}\gamma_i\). Naturally, this implies that \(\run((\vec{\sigma}_C,\vec{\sigma}_{-C}))\models\gamma_i\) for all \(i \in C\). But then this implies that \(\vec{\sigma}_C^\prime\) is a beneficial deviation for \(C\), contradicting the fact that \(\vec{\sigma}\) is a member of the core. Thus, we conclude that \(C\) and \(\Winners(\rho(\vec{\sigma}))\) have non-empty intersection.
\end{proof}

\begin{proof}[Proof of~\ref{lem:fc}.3]
  Let \(G\) be a game, and \(C\) be a fulfilled coalition. Let \(\vec{\sigma}^1 =(\vec{\sigma}_C, \vec{\sigma}_{-C})\) be any strategy profile under which \(C\) are fulfilled (that is, every member of \(C\) achieves their goal regardless of what the countercoalition does). If \(\vec{\sigma}^1\) is a member of the core, then we are done. If not, then there exists some coalition \(D \subseteq \Ag \setminus C\) with a beneficial deviation, \(\vec{\sigma}_D^\prime\). Then under the strategy \(\vec{\sigma}^2 = (\vec{\sigma}_C, \vec{\sigma}_D^\prime, \vec{\sigma}_{-(C \cup D)})\), every player in \(C\) gets their goal achieved. As before, if \(\vec{\sigma}^2\) is a member of the core, then we are done. If not, we can repeat this process, yielding a sequence of strategies, \(\vec{\sigma}^1, \vec{\sigma}^2, \vec{\sigma}^3 \ldots\). As the losers of the successive strategy profiles gets strictly smaller each time, this process can only continue a finite number of times. By construction, the final member of the sequence will model the goal of every player in \(C\), and will have no beneficial deviations, making it a member of the core. Thus, we see that there exists some strategy profile in the core such that \(C\) is a subset of the winners of the strategy profile.
\end{proof}

In fact, our proof of~\ref{lem:fc}.3 can be modified to give us a stronger result, namely that the core is always non-empty, a highly desirable game-theoretic property, as it ensures the existence of stable strategy profiles for every game, making them rationally implementable in practice. This contrasts with the conventional formulation of the core in transferable utility games~\cite{Chalkiadakis2011}.

\begin{theorem}\label{thm:emptiness}
    For every game~\(G\), we have \(\coreset{G}\neq\emptyset\).
  \end{theorem}

  \begin{proof} 
    Identical to the proof of Lemma~\ref{lem:fc}.3, but instead of setting \(\vec{\sigma}^1\) to be a strategy profile under which some coalition can be fulfilled, let it be any arbitrary strategy profile, and continue as before.
    \end{proof}

As fulfilled coalitions can help us understand the coalition formation power in a game, we will also be interested in the following decision problem about coalitions.
\begin{quote}
    \underline{\fcoalition}:\\
    \emph{Given}: Game \(G\), coalition \(C\subseteq\Ag\).\\
    \emph{Question}: Is \(C\) a fulfilled coalition of~\(G\)?
\end{quote}

In the next section, we will investigate these decision problems, as well as some model-theoretic properties of the \core.

\subsection{Characterising the Core}\label{secn:results-core1}
In this section we will study the complexity of the decision problems introduced in the previous section, and will establish some other properties of the \core. We have already seen in Theorem~\ref{thm:emptiness} that the core is always non-empty, and we shall go on to prove that the satisfaction of an LTL property on some/every outcome in the \core\ is bisimulation-invariant~\cite{Hennessy1985}. These two results sharply contrast with the rational verification problem in the non-cooperative setting: in these, the set of Nash equilibria of a game is not guaranteed to always be non-empty~\cite{Gutierrez2015}, nor does bisimulation-invariance hold in the general case~\cite{Gutierrez2017b}.

The first of our own decision problems we will consider is \fcoalition, which we solve in the general case through a logical characterisation using \atlstar.

\begin{theorem}\label{thm:fc}
    \fcoalition\ is \textsc{PSpace}-complete for one-player games, and it is \textsc{2ExpTime}-complete for games with more than one player.
\end{theorem}
\begin{proof}
    For membership we observe that given a game \(G=(M,\gamma_1,\ldots,\gamma_n)\) and a coalition~\(C\subseteq\Ag\), it is the case that~\(C\) is fulfilled if and only if \(s^0\models\ATLEpath{C}\bigwedge_{i\in C}\gamma_i\) holds. By appealing to \textsc{\text{ATL}\(^*\) Model Checking}, the two upper bounds immediately follow. For the lower bounds, we can reduce the problem of checking for the existence of a winning strategy in a two-player game with LTL goals as defined in~\cite{Alur2004} for \textsc{2ExpTime}-hardness and existential LTL model checking for \textsc{PSpace}-hardness~\cite{SistlaC85}. Note that, as such, \fcoalition\ with more than one player is \textsc{2ExpTime}-hard even for two-player zero-sum games, i.e., for games with $\gamma_1 = \neg\gamma_2$. 
\end{proof}

Fulfilled coalitions give an indication of which stable coalitions may form, but are insufficient to characterise the \core, and therefore, to check \ecore\ and \acore\ properties of a multi-agent system. To do this, we adopt a different strategy and show that these two decision problems are, in general, also \textsc{2ExpTime}-complete.

\begin{theorem}\label{thm:eacore}
    \ecore\ and \acore\ are \textsc{PSpace}-complete for one-player games and \textsc{2ExpTime}-complete for games with more than one player.
\end{theorem}
\begin{proof}
    Let us consider \ecore\ first. For membership we observe that given a game \(G\) and an LTL formula~\(\varphi\), it is the case that \((G,\varphi)\in\ecore\) if and only if \(s^0 \models\varphi_{\ecore}(G,\varphi)\) holds, where \(\varphi_{\ecore}(G,\varphi)\) is the following \atlstar\ formula:
    \begin{equation*}
        \bigvee_{W\subseteq\Ag}  \left(  \ATLEpath{\Ag}   \left(\varphi\wedge \bigwedge_{i\in W}\gamma_i \wedge \bigwedge_{j\in \Ag\setminus W} \neg\gamma_j\right)  \wedge  \bigwedge_{L\subseteq\Ag\setminus W}\ATLApath{L}\bigvee_{j\in L}\neg\gamma_j  \right)
    \end{equation*}
    which states that there is a path in \(G\) that satisfies \(\varphi\) as well as the goals of a set of players \(W\) (the ``winners''), and that for every subset of players \(L\) that do not get their goals achieved in such a path (the ``losers''), it is not the case that those players have a beneficial deviation from the path. As before, we appeal to \textsc{\text{ATL}\(^*\) Model Checking} to obtain the upper bounds; here we need to call a \textsc{2ExpTime} algorithm an exponential number of times---one for each coalition of winners, and accept if any single one of them accepts.

For the lower bounds, as for \fcoalition, we can reduce the problem of checking for the existence of a winning strategy in a two-player game with LTL goals as defined in~\cite{Alur2004} for \textsc{2ExpTime}-hardness and existential LTL model checking for \textsc{PSpace}-hardness~\cite{SistlaC85}. In fact, note that for one-player games, \ecore\ and \fcoalition\ are equivalent when $\phi=\gamma_1$, as well as when $\gamma_1=\phi$, $\gamma_2=\neg\phi$, and $C=\set{1}$, which are the cases that arise in the two reductions above mentioned. 

    Finally, for \acore, note that \((G,\varphi)\not\in\acore\) if and only if \((G,\neg\varphi)\in\ecore\): \((G,\varphi)\not\in\acore\) means it is not the case that for all members of the core \(\vec{\sigma}\) we have \(\rho(\vec{\sigma}) \models \varphi\). Pushing the negation through the quantifier, this implies that there exists a member of the core such that \(\rho(\vec{\sigma}) \not\models \varphi\), or \(\rho(\vec{\sigma}) \models \lnot\varphi\). The same argument can then be applied for the reverse direction. Then, since both \textsc{PSpace} and \textsc{2ExpTime} are deterministic complexity classes, we can conclude that \acore\ is \textsc{PSpace}-complete if \(|\Ag|=1\) and \textsc{2ExpTime}-complete if \(|\Ag|>1\), as it is for \ecore.
\end{proof}

We now study \cmembership\ and \bdeviation. For these two problems we first need to define how we will represent strategy profiles, as at present, strategies are defined as infinitary structures, which map finite histories to players' actions. 

Given this finite representation for strategies, we can establish the complexity of \cmembership\ and \bdeviation.
\begin{theorem}\label{thm:cmembership}
    \cmembership\ is \textsc{PSpace}-complete for one-player games and \textsc{2ExpTime}-complete for games with more than one player.
\end{theorem}
\begin{proof}
    For membership, we first compute the winners and losers with respect to \(\vec{\sigma}=(\sigma_1,\ldots,\sigma_n)\), the outcome of the game. This can be done in \textsc{PSpace} (it is equivalent to LTL model checking over a ``product automata'' or ``concurrent program''~\cite{Kupferman2000}). Once we have computed \(W\), we can check, for every \(L\subseteq\Ag\setminus W\), whether \(L\) has a beneficial deviation. This is true if and only if \(L\) is a fulfilled coalition. Because this can be checked in \textsc{PSpace} for one-player games and in \textsc{2ExpTime} for games with more than one player, the two upper bounds immediately follow. For the lower bounds, we use Lemma~\ref{lem:fc} and Theorem~\ref{thm:fc} again. Consider the following game. Let \(\varphi\) be a satisfiable LTL formula and \(\vec{\sigma}\) an outcome that does not satisfy \(\varphi\). Then, \((G,\vec{\sigma})\in\cmembership\) if and only if \((G,\set{1})\not\in\fcoalition\), whenever \(\gamma_1=\varphi\) and \(\gamma_j=\neg\varphi\), for every player~\(j\in\Ag\setminus\set{1}\).%
\end{proof}

Let us now consider \bdeviation. This is the only ``easy'' problem for multi-player games: it can be solved in \textsc{PSpace}\@. To show this, we again need to find a different proof strategy. Consider any input instance \((G,\vec{\sigma},\vec{\sigma}_{C}')\) of the problem. We observe that, because \(\vec{\sigma}_{C}'\) is fixed, we can make it part of the arena where the game is played, and then check if players not in \(C\) have a joint strategy for \(\bigvee_{j\in C}\neg\gamma_j\). Due to the definition of beneficial deviation, we also need to check if \(\run(\vec{\sigma})\models\bigwedge_{j\in }\neg\gamma_j\) holds or not.

In other words, the reason why this problem can be solved in \textsc{PSpace} for multi-player games, unlike all other decision problems we have studied so far (which, in general, can be solved in doubly exponential time), is that this decision problem can be reduced to a one-player game (given by coalition \(\Ag\setminus C\)) with an LTL goal (given by \(\gamma_{\Ag\setminus C}=\bigvee_{j\in C}\neg\gamma_j\)) over a ``product arena'' (denoted by \(M_{C}\)) built from a concurrent game structure \(M\) and the joint strategy \(\vec{\sigma}_{C}'\) that we want to check.

\begin{theorem}\label{thm:bdeviation}
    \bdeviation\ is \textsc{PSpace}-complete, even for one-player games.
\end{theorem}
\begin{proof}
    Checking that \(\run(\vec{\sigma})\models\bigwedge_{j\in C}\neg\gamma_j\) holds can be done in \textsc{PSpace}\@. Again, this is equivalent to model checking LTL formulae over a ``product automata'' or ``concurrent program''~\cite{Kupferman2000}. If the statement does not hold, then, by definition, \(\vec{\sigma}'_{C}\) is not a beneficial deviation, as at least one player in~\(C\) already has its goal satisfied by \(\vec{\sigma}\). If the statement holds, then we check that \(\run(\vec{\sigma}'_{-C},\vec{\sigma}'_{C})\models\bigwedge_{j\in C}\gamma_j\) holds, for all joint strategies~\(\vec{\sigma}'_{-C}\) for players not in \(C\). We do this in \textsc{PSpace} by checking whether it is not the case that \((M_{C}, \valf', s^{0'}) \models\bigvee_{j\in C}\neg\gamma_j\) holds, where \(M_{C}=(\Ag',\Ac',\St',s^{0'},\transf')\) is the concurrent game structure defined as follows:
    \begin{itemize}
        \item \(\Ag'=\set{0}\), \(\Ac' = \Pi_{i\in \Ag\setminus C} \Ac_i\);
        \item \(\St' = \St \times \Pi_{j \in C}Q_j\);
        \item \(s^{0'} = (s^0,q^0_x,\ldots,q^0_y)\), such that \(\sigma_z=(Q_z,q^0_z,\strtransf_z,\stroutf_z)\), \(\vec{\sigma}'_{C} = (\sigma_x,\ldots,\sigma_y)\), and \(z\in\set{x,\ldots,y}\);
        \item \(\transf'((s,q_x,\ldots,q_y),(a,\ldots,b)) = (s',q'_x,\ldots,q'_y)\) such that
              \begin{itemize}
                  \item \(s' = \transf(s,\stroutf(q_x),\ldots,\stroutf(q_y),a,\ldots,b)\), and
                  \item \(q'_z = \strtransf(q_z,s)\), with \(z\in\set{x,\ldots,y}\).
              \end{itemize}
    \end{itemize}
    and \(\valf'\) is defined as,
    \begin{equation*}
        \valf'(s,q_x,\ldots,q_y) = \valf(s).
    \end{equation*}
    In other words, \(M_{C}\) transitions just like \(M\) save that it is restricted to the behaviour already defined by \(\vec{\sigma}'_{C}\).

    For the lower bound we use LTL model checking. 
\end{proof}

In addition to the above complexity results, we also have a model-theoretic result. Before we can state it, however, we need to define the notion of bisimilarity~\cite{Milner1989,Hennessy1985}. Let \(M\) and \(M^\prime\) be two concurrent game structures with the same agents \(\Ag\), actions \(\Ac\) and atomic propositions \(\AP\). Moreover, let their respective sets of states and transition functions be denoted by \(\St/\St^\prime\) and \(\transf/\transf^\prime\) respectively. Finally, let \(\lambda\) be a labelling function on \(\St\) and \(\lambda^\prime\) a labelling function of \(\St^\prime\). Then a \emph{bisimulation} between \(s^* \in \St\) and \(t^* \in \St^\prime\) is a non-empty binary relation, \(\sim \subseteq \St \times \St^\prime\), such that,
\begin{itemize}
\item We have \(s^* \sim t^*\);
\item For all \(s \in \St\) and for all \(t \in \St^\prime\), if \(s \sim t\) then \(\lambda(s) = \lambda^\prime(t)\);
\item For all \(s_1, s_2 \in \St\) and for all \(t_1 \in \St^\prime\), if \(s_1 \sim t_1\) and \(\transf(s_1, \ac) = s_2\) for some \(\ac \in \Ac\), then \(\transf^\prime(t_1, \ac) = t_2\) for some \(t_2 \in \St^\prime\) with \(s_2 \sim t_2\);
  \item For all \(s_1 \in \St\) and for all \(t_1, t_2 \in \St^\prime\), if \(s_1 \sim t_1\) and \(\transf(t_1, \ac) = t_2\) for some \(\ac \in \Ac\), then \(\transf^\prime(s_1, \ac) = s_2\) for some \(s_2 \in \St^\prime\) with \(s_2 \sim t_2\).
  \end{itemize}

We say two concurrent game structures are bisimilar if their start states are bisimilar. For further details of bisimilarity over concurrent game structures, refer to~\cite{Gutierrez2017b}. We then say that a property \(P\) is bisimulation-invariant when if \(M \sim N\), then \(P\) holds on \(M\) if and only if \(P\) holds on \(N\). 

We are now in a position to state our result: informally, we have that checking whether an LTL formula is satisfied by some outcome in the \core\ is a bisimulation-invariant property. This result is easy, and follows directly from the membership proof of \ecore.

\begin{corollary}\label{cor:bisim}
    Let \(G=(M,\gamma_1,\ldots,\gamma_n)\) be a game, \(\varphi\) be an LTL formula, and~\(M'\) be a concurrent game structure that is bisimilar to \(M\). Then, \((G,\varphi)\in\ecore\) if and only if \((G',\varphi)\in\ecore\), where \(G'=(M',\gamma_1,\ldots,\gamma_n)\).
\end{corollary}
\begin{proof}
    Follows from the fact that \atlstar\ is bisimulation-invariant, and that the \core\ can be characterised in \atlstar\ using \(\varphi_{\ecore}\), as defined in the proof of Theorem~\ref{thm:eacore}. More specifically, it follows from the fact that \((M, \lambda, s^0)\models\varphi_{\ecore}(G,\varphi)\) if and only if \((M', \lambda', s^{0\prime})\models\varphi_{\ecore}(G',\varphi)\).
  \end{proof}  

\subsection{On Credible Coalition Formation: The Strong Core}\label{secn:strong-core}

As we noted above, our definition of the core assumes worst-case reasoning: a deviation must be beneficial against \emph{all} counter-responses. This definition is robust in the sense that any core-stable outcome is stable in a very strong sense, but one could argue that in some cases it is \emph{too} strong. In particular, when a coalition \(C\) is contemplating a deviation \(\vec{\sigma}_{C}\), it can surely assume that the remaining players will not act against their own interests. Thus, one could argue that a deviation need not be beneficial for \emph{all} behaviours of the remaining players, but only those behaviours that are \emph{credible}, in the sense that the remaining players might rationally choose them. To make this discussion concrete, consider the following example.
\begin{example}\label{ex:example-2}
  Suppose we have a two-player game \(G\), with a start state, \(s_0\) and three sink states, \(s_1, s_2\), and \(s_3\). Each player has two actions available to them, \(a\) and \(b\), and the transition function from the start state is defined as follows:
    \begin{align*}
        \transf(s_0, (a, a)) & = s_1, \\
        \transf(s_0, (a,b))  & = s_2, \\
        \transf(s_0, (b, a)) & = s_3, \\
        \transf(s_0, (b,b))  & = s_3
    \end{align*}
    Thus the game is as illustrated in Figure~\ref{fig:non-credible}. Additionally, suppose that the infinite run that ends up in \(s_1\) is the one run which satisfies player one's goal, and the run that ends up in \(s_2\) is the one which satisfies player two's goal.  Now, in this game, the run which ends up in \(s_1\) lies in the core, but with the use of a non-credible (punishing) strategy by player 1. Notice that the only possible deviation from \((a, a)\) for player 2 is to play \(b\), to which player 1 could respond by also playing \(b\). Although this behaviour would prevent player 2 from achieving its goal, such a way of playing can be regarded as not rational for player 1 given their preference relation: player 1 certainly prefers the run which ends in \(s_1\) over the other two possible runs, but is indifferent otherwise.

    \begin{figure}[H]\label{fig:non-credible}
        \centering
        \begin{tikzpicture}[state/.style={circle, draw, minimum size=1cm}, every node/.append style={transform shape}, node distance=1.5cm]
			\node		 (start) {} ;
			\node[state] (s_0) at ($(start) + (90:1.5cm) $) {\(s_0\)};
			\node[state] (s_1) [right = of s_0] {\(s_1\)};
			\node[state] (s_2) [above = of s_0] {\(s_2\)};
			\node[state] (s_3) [left = of s_0] {\(s_3\)};

			\draw [-{Latex[width=2mm]}]
			(start) edge[] (s_0)
			(s_0) edge[] node[above] {\((a, a)\)}  (s_1)
			(s_0) edge[] node[right] {\((a, b)\)}  (s_2)
			(s_0) edge[] node[above] {\((b, b)\)} node[below] {\((b, a)\)} (s_3)

			(s_1) edge[loop, out=20, in=340, distance=1cm] node[right]{\(*\)} (s_1)
			(s_2) edge[loop, out = 110, in=70, distance=1cm] node[above] {\(*\)} (s_2)
			(s_3) edge[loop, out=200, in=160, distance=1cm] node[left]{$*$} (s_3)
			;
\end{tikzpicture}
        \caption{A game with a non-credible strategy.}
    \end{figure}
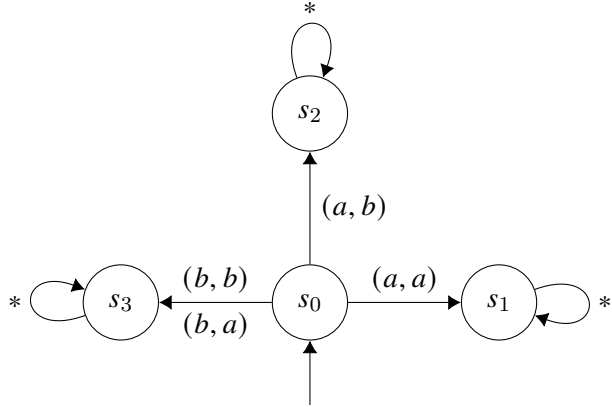
\end{example}
Motivated by this phenomenon, we now present a stronger definition for the \core\@. More specifically, with this new definition we require that if a coalition \(C\) wants to deviate from a given strategy profile, then the remaining players can only credibly threaten \(C\) when they have a counter-response in which both at least one player in \(C\) does not get its goal achieved and every winner in the original strategy profile remains a winner in the new one, \emph{i.e.}, the counter-coalition act in accordance with their preference relations. In this solution concept, we are thus capturing the idea of remaining players being willing to punish deviators, but only up to a point: those left behind would prefer not to do worse than they were doing originally. We do not claim that this solution concept is always appropriate, as is the case with all solution concepts in cooperative games with externalities. But it captures another interesting variation of how a non-deviating coalition might respond to a deviation.

We then reformulate the definition of a beneficial deviation, and now say that a deviation \(\vec{\sigma'}_{C}\) is a \emph{strong beneficial deviation} from \(\vec{\sigma}\) if the following conditions hold:
\begin{enumerate}
    \item \(C \subseteq \Losers(\vec{\sigma})\);
    \item \(C\subseteq\Winners(\vec{\sigma}_{-C},\vec{\sigma'}_{C})\);
    \item For every joint strategy \(\vec{\sigma}_{-C}'\) for \(\Ag\setminus C\), we have that if \(\Winners(\vec{\sigma})\subseteq\Winners(\vec{\sigma}_{-C}',\vec{\sigma}_{C}')\), then \(C\subseteq\Winners(\vec{\sigma}_{-C}',\vec{\sigma}_{C}')\).
\end{enumerate}

With this definition in place we say that the \textbf{strong core} of a game, denoted \(\ccoreset{G}\), is the set of outcomes of \(G\) that admit no strong beneficial deviation. Then, with respect to Example~\ref{ex:example-2}, we see that while~\(\vec{\sigma}_{a_1a_1}\) is in \(\coreset{G}\), it is not the case that~\(\vec{\sigma}_{a_1a_1}\) is in \(\ccoreset{G}\), since player~2 can now strongly beneficially deviate from \(\vec{\sigma}_{a_1a_1}\) to \(\vec{\sigma}_{a_1b_1}\).

It is worth pausing to reflect for a moment on the issue of formulating the core in the presence of externalities. The game theory literature on this topic is very large. The reason is that the existence of externalities leads to many different definitions of stable behaviour (see, \emph{e.g.},~\cite{OsborneR94,yi1997stable,uyanik2015nonemptiness,Finus2003} for many variants of the \core). Here, we propose one definition but by no means do we claim it is the only possibility. Essentially, with our definition, we require that for a punishing joint strategy to be credible, winners must remain winners after the presenting the threat.

We will now study the complexity of the decision problems defined in previous sections, but with respect to the strong core. There are four decision problems whose definition depends on the nature of the \core: \ecore, \acore, \cmembership, and \bdeviation. We will use the same names for these problems, with the understanding that results in this section are with respect to the strong core. As we will show next, these four problems have the same complexities as with \core, but require a more complex logical characterisation, which we provide here using the two-alternation fragment of Strategy Logic (SL)~\cite{MogaveroMPV14}.\footnote{We were unable to find a logical characterisation of the strong core using \atlstar\@. In fact, we believe that such a logical characterisation in \atlstar\ is not possible for multi-player games.}

\begin{theorem}\label{thm:core-membership}
    Given a game \(G=(M,\gamma_1,\ldots,\gamma_n)\) and LTL formula \(\varphi\), we have \((G,\varphi)\in\ecore\) if and only if \(M\models\varphi^+_{\ecore}(G,\varphi)\), where \(\varphi^+_{\ecore}(G,\varphi)\) is the SL formula:
    \begin{align*}
        \varphi^+_{\ecore}(G,\varphi) & =  \bigvee_{W\subseteq\Ag}\SLE{\Ag}\left( \varphi \wedge \bigwedge_{i\in W} \gamma_i \wedge \bigwedge_{j\in\Ag\setminus W} \neg\gamma_j \wedge \bigwedge_{C\subseteq\Ag\setminus W} \varphi_{\mathrm{NoBD}}(G,W,C) \right)
    \end{align*}
    and the formula \(\varphi_{\mathrm{NoSBD}}(G,W,C)\) is defined as follows:
    \begin{align*}
        \varphi_{\mathrm{NoSBD}}(G,W,C) & = \SLA{C} \left(\bigwedge_{j\in C} \gamma_j \rightarrow \SLE{\Ag\setminus C} \left( \bigwedge_{i\in W}\gamma_i \wedge \bigvee_{j\in C}\neg\gamma_j \right)\right)
    \end{align*}
\end{theorem}
\begin{proof}
    This SL formula expresses that in the concurrent game structure, there exists a path \(\SLE{\Ag}\left(\ldots\right)\) under which
    \begin{enumerate}
        \item The formula \(\varphi\) holds;
        \item Some players get their goals achieved: \(\bigwedge_{i\in W}\gamma_i\);
        \item The remaining players do not: \(\bigwedge_{j\in \Ag\setminus W}\neg\gamma_j\);
        \item No coalition of losers has a strong beneficial deviation: \(\bigwedge_{C\subseteq\Ag\setminus W} \varphi_{\mathrm{NoSBD}}(G,W,C)\).
    \end{enumerate}

    We express the condition of a coalition of losers \(C\) not having a strong beneficial deviation with the SL formula \(\varphi_{\mathrm{NoSBD}}(G,W,C)\); this is broken down as follows: for every joint strategy of~\(C\), if every player in \(C\) is better off \(\left(\bigwedge_{j\in C}\gamma_j\right)\), then the coalition of players outside \(C\) have a joint strategy \(\left(\SLE{\Ag\setminus C}\ldots\right)\) such that both the winners in the original outcome remain winners after the threat is presented \(\left(\bigwedge_{i\in W}\gamma_i\right)\), and at least one player in the deviating coalition, \(C\), does not get its goal achieved \(\left(\bigvee_{j\in C}\neg\gamma_j\right)\).
\end{proof}
At this point, we would like to make a couple of observations. First, that the complexity of checking SL formulae is non-elementary and depends on the alternation-depth of the formula (\cite{MogaveroMPV14}): SL formulae of alternation-depth~\(n\) can be checked in \((n+1)\)-\textsc{ExpTime}, and in \textsc{PSpace} for formulae that are semantically equivalent to CTL\(^*\) formulae. Since \(\varphi^+_{\ecore}(G,\varphi)\) is an SL formula with two alternations, it can be checked in \textsc{3ExpTime} (and in \textsc{PSpace} if \(|\Ag|=1\)). Second, we also would like to recall that finite-state machine strategies, as those we use here, can be characterised in LTL using the technique presented in~\cite{Gutierrez2015,Gutierrez2017a}. Using these logical characterisations, we can obtain the following complexity results. 

\begin{theorem}\label{thm:credible}
    For multi-player games, while \ecore\ and \acore\ are in \textsc{3ExpTime}, \cmembership\ is \textsc{2ExpTime}-complete and \bdeviation\ is \textsc{PSpace}-complete. For one-player games, all problems are \textsc{PSpace}-complete.
\end{theorem}

%
Because we characterised the strong core using SL (which, in contrast to \atlstar, is not a bisimulation-invariant logic), we cannot conclude that the satisfaction of LTL properties by outcomes in the strong core is a bisimulation-invariant property. We believe that this is not the case.

In our first formulation of the core, we saw that the core is always non-empty. Thus, a natural question is to ask whether the strong core is always non-empty. We begin by showing that this is the case for games with three or fewer players:

\begin{theorem}\label{thm:strong-core-nonempty}
    For every two-player and three-player game, \(G\), we have \(\ccoreset{G} \neq\emptyset\).
\end{theorem}
\begin{proof}
    First, consider games with only two players. For a contradiction, let us suppose that for some game \(G\), the set of outcomes \(\ccoreset{G}\) is empty. This means that for every outcome either player~1 or player~2 or both have a strong beneficial deviation. Then, we know that no outcome can satisfy both goals, \(\gamma_1\) and \(\gamma_2\). Let us then consider the three remaining possible cases: outcomes that only satisfy \(\gamma_1\) (case~1), outcomes that only satisfy \(\gamma_2\) (case~2), and outcomes that satisfy neither \(\gamma_1\) nor \(\gamma_2\) (case~3). Let \(\vec{f}=(f_1,f_2)\) be an outcome, \(f'_1\) be a deviation by player~1, and \(f'_2\) be a deviation by player~2, and consider the three cases above. In case~1, only player~2 would deviate. Then, outcome \((f_1,f'_2)\) only satisfies \(\gamma_2\). Because \((f_1,f'_2)\) is not in the core either, from this outcome only player~1 would deviate, to another outcome \((f'_1,f'_2)\). Then, outcome \((f'_1,f'_2)\) only satisfies \(\gamma_1\). But, then, we have a contradiction, since this means that \((f_1,f_2)\) would be in \(\ccoreset{G}\). We can reason symmetrically to show that case~2 is not possible either. For case~3 we note that only single deviations would be possible. But any such deviations would be to an outcome that either only satisfies \(\gamma_1\) or only satisfies \(\gamma_2\), which are no longer possible. Since no other cases are possible, we have to reject our assumption and conclude that, for two-player games, \(\ccoreset{G}\) is not empty.

    For three-player games, the proof is similar, but requires a careful case-by-case analysis. For a full proof, see the appendix.
\end{proof}

In contrast, for games with four or more players, we find that the strong core may be empty, as the following example illustrates.

\begin{example}\label{ex:empty-strong-core}
    Consider the following 4 player game, with a start state, \(s^0\), and six sink states, \(s^1, \ldots, s^6\). Let each player have two actions each, \(\{0,1\}\), and writing \(a_1a_2a_3a_4\) for the action \((a_1, a_2, a_3, a_4)\), our transition function from the start state looks like the following:
    \begin{align*}
        \transf(0000) & = s^1, \quad                 &  & \transf(0001)  = s^1,         \\
        \transf(0010) & = s^2, \quad                 &  & \transf(0011)  = s^2,         \\
        \transf(0100) & = s^1, \quad                 &  & \transf(0101)  = s^3,         \\
        \transf(0110) & = s^2, \quad                 &  & \transf(0111)          = s^5, \\
        \transf(1000) & = s^6,     \quad             &  & \transf(1001)  = s^4,         \\
        \transf(1010) & = s^4,         \quad         &  & \transf(1011)  = s^4,         \\
        \transf(1100) & = s^1,        \quad          &  & \transf(1101)  = s^3,         \\
        \transf(1110) & = s^4,                 \quad &  & \transf(1111)  = s^3.         \\
    \end{align*}
    Moreover, suppose player one prefers the runs that end up in the states \(s^1, s^2, s^3\), player two prefers those that end up in \(s^1, s^4, s^5\), player three \(s^2, s^4, s^6\), and player four \(s^3,s^5, s^6\).
    The game is illustrated in~Figure\ref{fig:empty-ccore-example} --- we label states with the players that are winners in that state, and edges with the possible action tuples.\footnote{We would like to note that this (counter-)example was automatically generated using bounded exhaustive search of games of different size, which may explain why the game is so counter-intuitive, and most importantly why a game like this one is so hard to find, or be produced, by hand.}
    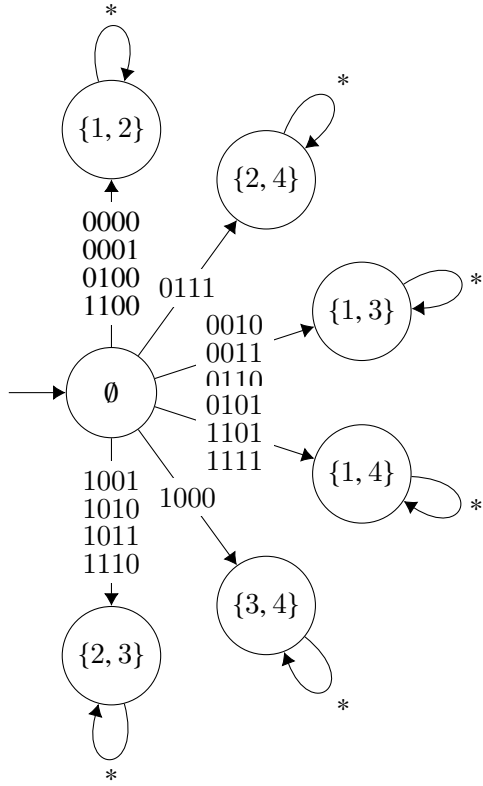
\begin{figure}[H]
        \centering
        \begin{tikzpicture}[state/.style={circle, draw, minimum size=1.2cm}, every node/.append style={transform shape}]
			\node (start) {};
			\node[state] (s_0) at ($(start) + (0:1.5cm) $)  {\(\emptyset\)};
			\node[state] (s_1) at ($(s_0) + (90:3.5cm) $) {\(\{1,2\}\)};
			\node[state] (s_5) at ($(s_0) + (54:3.5cm) $) {\(\{2,4\}\)};
			\node[state] (s_2) at ($(s_0) + (18:3.5cm) $) {\(\{1,3\}\)};
			\node[state] (s_3) at ($(s_0) + (-18:3.5cm) $) {\(\{1,4\}\)};
			\node[state] (s_6) at ($(s_0) + (-54:3.5cm) $) {\(\{3,4\}\)};
            \node[state] (s_4) at ($(s_0) + (-90:3.5cm) $) {\(\{2,3\}\)};

			\draw [-{Latex[width=2mm]}]
			(start) edge[] (s_0)
			(s_0) edge[] node[align=left, fill=white]{ 0000 \\[-2ex] 0001 \\[-2ex] 0100 \\[-2ex] 1100}  (s_1)
			(s_0) edge[] node[align=left, fill=white] {\(0010\) \\[-2ex] \(0011\) \\[-2ex] \(0110\) }  (s_2)
			(s_0) edge[] node[align=left, fill=white] {\(0101\) \\[-2ex] \(1101\) \\[-2ex] \(1111\) }  (s_3)
			(s_0) edge[] node[align=left, fill=white] {\(1001\) \\[-2ex] \(1010\) \\[-2ex] \(1011\) \\[-2ex] \(1110\)}  (s_4)
			(s_0) edge[] node[align=left, fill=white] {\(0111\)}  (s_5)
			(s_0) edge[] node[align=left, fill=white] {\(1000\)}  (s_6)

			(s_1) edge[loop, out=105, in=75, distance=1cm] node[above] {\(*\)} (s_1)
			(s_5) edge[loop, out=69, in=39,  distance=1cm] node[above right]{\(*\)} (s_5)
			(s_2) edge[loop ,out=33, in=3, distance=1cm] node[right]{\(*\)} (s_2)
			(s_3) edge[loop, out=-3, in=-33,  distance=1cm] node[right]{\(*\)} (s_3)
			(s_6) edge[loop, out=-39, in=-69,  distance=1cm] node[below right]{\(*\)} (s_6)
			(s_4) edge[loop, out=-75, in=-105,  distance=1cm] node[below]{\(*\)} (s_4)
			;
\end{tikzpicture}
        \caption{A game with an empty \ccore.}\label{fig:empty-ccore-example}
    \end{figure}
    With a careful case-by-case analysis, one can verify that for every state, there exists some coalition with a strong beneficial deviation. Thus, the strong core of the game is empty. For a complete analysis of the deviations that each coalition can make, please refer to the appendix where the full case-by-case analysis is presented.
\end{example}

\section{Mean-Payoff Games}\label{secn:mean-payoff-games}

Thus far, we have considered games with \emph{qualitative} preferences---each player has a goal given by some temporal logic formula, which, under a given run, is either satisfied or unsatisfied. But this is a very coarse-grained approach to specifying preferences; it does not offer a way of expressing the \emph{intensity} of the individual player's preferences. One possible way to obtain a richer model of preferences would be to introduce multiple LTL goals for each player, and define some mapping from the set of satisfied formulae to the real numbers~\cite{Mavronicolas2007,Almagor2018,Kupferman2016}. However, an alternative approach, which has been widely studied in the literature, is to sidestep temporal logics entirely and assign weights to states, rather than atomic propositions. We then compute the \textbf{mean-payoff} of runs, with the idea that agents prefer runs which maximise their mean-payoff~\cite{Ehrenfeucht1979,Zwick1996,Ummels2011}. In this section, we will revisit the formulation of the core in this mean-payoff setting.

Formally, a mean-payoff game, \(G\), is a tuple,
\begin{equation*}
    G = (M, {\{w_i\}}_{i \in \Ag}),
\end{equation*}
where \(M\) is a concurrent game structure, and for each \(i \in \Ag\), \(w_i : \St \to \mathbb{Z}\) is a \textbf{weight} function, mapping states to integers. Games are played in an identical way to the LTL setting, but the agents' preference relations are defined differently here. Let \(\beta \in \mathbb{R}^\omega\) be an infinite sequence of real numbers. Then the mean-payoff of \(\beta\), denoted by \(\mp(\beta)\), is defined as follows:
\begin{equation*}
    mp(\beta) = \liminf_{n\to\infty} \frac{1}{n} \sum_{i=0}^{n-1} \beta_i.
\end{equation*}
In a mean-payoff game, a run \(\rho = s^0 s^1 \ldots\) induces an infinite sequence of weights for each player, \(w_i(s^0)w_i(s^1)\ldots\) --- we denote this sequence by \(w_i(\rho)\) and for notational convenience, we will write \(\pay_i(\rho)\) for \(\mp(w_i(\rho))\). With this, we can define the preference relation for each player: given two runs, \(\rho\) and \(\rho'\), we have \(\rho \succeq_i \rho'\) if \(\pay_i(\rho) \geq \pay_i(\rho')\); the strict relation \(\succ_i\) is defined in the usual way.

Now, recall that our definition of the core in the setting of LTL games relies on the notion of winners and losers of a game. In the mean-payoff setting, it clearly makes no sense to classify players as winners or losers---they can receive a wide spectrum of payoffs. Thus, we need to revisit the concept of a beneficial deviation. In the mean-payoff setting, we say that given a strategy profile \(\vec{\sigma}\), a beneficial deviation by a coalition \(C\) is a strategy vector \(\vec{\sigma}_C'\) such that for all complementary strategy profiles \(\vec{\sigma}_{\Ag \setminus C}'\), we have \(\rho(\vec{\sigma}_C', \vec{\sigma}_{\Ag \setminus C}') \succ_i \rho(\vec{\sigma})\) for all \(i \in C\). We then say that \(\vec{\sigma}\) is a member of the core if there exists no coalition \(C\) which has a beneficial deviation from \(\vec{\sigma}\).

\subsection{Non-Emptiness of the Core}
We begin by asking whether mean-payoff games always have a non-empty core. Recall that in the LTL games case, we found the core was guaranteed to be non-empty: we find that this does not hold in general for mean-payoff games. Before proving this theorem, we need a small lemma showing that\ldots

\begin{lemma}\label{lmm:mean-payoff-values-are-closed}
  Let \({\{\vec{\sigma}^j\}}_{j \in \mathbb{N}}\) be a sequence of strategy profiles, and suppose that for some player \(i\), we have \(\lim_{j \to \infty}\pay_i(\rho(\vec{\sigma}^j)) = x\) for some \(x \in \mathbb{R}\). Then there exists some strategy profile \(\vec{\sigma}^x\) such that \(\pay_i(\rho(\vec{\sigma}^x)) = x\).
\end{lemma}

\begin{proof}
  First note that the game will end up in a strongly-connected component. So let \(\mathcal{C}\) be the set of all simple cycles of the game graph of said strongly-connected component. 
  Consider linear program with solution \(x\) --- this gives proportion of cycles 
\end{proof}

\begin{theorem}\label{thm:empty-core}
    In mean-payoff games, if \(\abs{\Ag}\leq 2\), then the core is non-empty. For \(\abs{\Ag} >2\), there exist games with an empty core.
\end{theorem}
\begin{proof}
    If \(\abs{\Ag}=1\), it is straightforward to see that the core is always non-empty; we use Karp's algorithm for determining the maximum cycle in a weighted graph~\cite{Karp1978} to determine the maximum payoff that one player can achieve. For two-player games, let \(\vec{\sigma}=(\sigma_1,\sigma_2)\) be any strategy profile. If \(\vec{\sigma}\) is not in the core, then either Player~1, or Player~2, or the coalition consisting of both players has a beneficial deviation. If the latter is true, then there is a strategy profile, \(\vec{\sigma}' = (\sigma_1', \sigma_2')\) such that \(\vec{\sigma}' \succ_i \vec{\sigma}\) for both \(i \in \{1,2\}\). We repeat this process until the coalition of both players does not have a beneficial deviation. This must eventually be the case as 1) each player's payoff is capped by their maximum weight and 2) by Theorem~4 of~\cite{Brenguier2015}, we see that the set of payoffs that a coalition can achieve is a closed set, so any limit point can be attained.\footnote{The result of~\cite{Brenguier2015} actually refers to two-player, multi-mean-payoff games, whilst we are working with multi-player mean-payoff games. We will cover why this is not an issue momentarily.} So there must come a point when they cannot beneficially deviate together. At this point, we must either be in the core, or either player 1 or player 2 has a beneficial deviation. If player \(j \in \{1,2\}\) has a beneficial deviation, say \(\sigma_j\), then any strategy profile \((\sigma_j,\sigma_i)\), with \(i\neq j\), that maximises Player~\(i\)'s mean-payoff is in the core. Thus, for every two-player game, there exists some strategy profile that lies in the core.

    However, for mean-payoff games with three or more players, the core of a game may be empty. The following example illustrates this case.

    \begin{example}Consider the following three-player game \(G\), where each player has two actions, \(\HEADS, \TAILS\), and there are four states, \(P, R, B, Y\). The states are weighted for each player as follows:

        \begin{center}
            \begin{tabular}{cccc}
                \toprule
                \(w_i(s)\) & \(1\)  & \(2\)  & \(3\)  \\
                \midrule
                \(P\)      & \(-1\) & \(-1\) & \(-1\) \\
                \(R\)      & \(2\)  & \(1\)  & \(0\)  \\
                \(B\)      & \(0\)  & \(2\)  & \(1\)  \\
                \(Y\)      & \(1\)  & \(0\)  & \(2\)  \\
                \bottomrule
            \end{tabular}
        \end{center}

        If the game is in any state other than \(P\), then no matter what set of actions is taken, the game will remain in that state. Thus, we only specify the transitions for the state \(P\):

        \begin{center}
            \begin{tabular}{cc}
                \toprule
                \(\Ac\)                      & \(\St\) \\
                \midrule
                \((\HEADS, \HEADS, \HEADS)\) & \(R\)   \\
                \((\HEADS, \HEADS, \TAILS)\) & \(R\)   \\
                \((\HEADS, \TAILS, \HEADS)\) & \(B\)   \\
                \((\HEADS, \TAILS, \TAILS)\) & \(P\)   \\
                \((\TAILS, \HEADS, \HEADS)\) & \(P\)   \\
                \((\TAILS, \HEADS, \TAILS)\) & \(Y\)   \\
                \((\TAILS, \TAILS, \HEADS)\) & \(B\)   \\
                \((\TAILS, \TAILS, \TAILS)\) & \(Y\)   \\
                \bottomrule
            \end{tabular}
        \end{center}

        Figure~\ref{fig:empty-core} illustrates the structure of the game.
        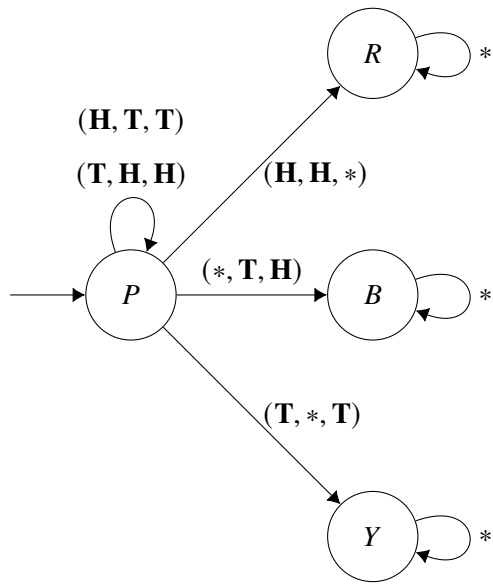
\begin{figure}[H]
            \centering
            \begin{tikzpicture}[state/.style={circle, draw, minimum size=1.2cm}, node distance=2cm]
		\node[state] (s^0) {$P$};
		\node[] (start) [left =1cm of s^0] {};
		\node[state] (s^2) [right = of s^0] {$B$};
		\node[state] (s^1) [above = of s^2] {$R$};
		\node[state] (s^3) [below = of s^2] {$Y$};

		\draw [-{Latex[width=2mm]}]
		(start) edge node{} (s^0)
		(s^0) edge[] node[right]{$(\HEADS, \HEADS, *)$} (s^1)
		(s^0) edge[] node[above]{$(*, \TAILS, \HEADS)$} (s^2)
		(s^0) edge[] node[right]{$(\TAILS, *, \TAILS)$} (s^3)

		(s^0) edge[loop, out=110, in=70, distance=1cm] node[above,align=center]{$(\HEADS, \TAILS, \TAILS)$\\$(\TAILS, \HEADS, \HEADS)$} (s^0)
		(s^1) edge[loop, out=20, in=340, distance=1cm] node[right]{$*$} (s^1)
		(s^2) edge[loop, out=20, in=340, distance=1cm] node[right]{$*$} (s^2)
		(s^3) edge[loop, out=20, in=340, distance=1cm] node[right]{$*$} (s^3)
		;
\end{tikzpicture}
            \caption{A game with an empty core.\label{fig:empty-core}}
        \end{figure}

        Note that strategies are characterised by the state that the game eventually ends up in. If the players stay in \(P\) forever, then they can all collectively change strategy to move to one of \(R, B, Y\), and each get a better payoff. Now, if the game ends up in \(R\), then players 2 and 3 can deviate by playing \((\TAILS, \HEADS)\), and no matter what player 1 plays, the game will be in state \(B\), leaving the two deviating players better off. But similarly, if the game is in \(B\), then players 1 and 3 can deviate by playing \((\TAILS, \TAILS)\) to enter state \(Y\), in which they both will be better off, regardless of what player 2 does. And finally, if in \(Y\), then players \(1\) and \(2\) can deviate by playing \((\HEADS, \HEADS)\) to enter \(R\) and will be better off regardless of what player \(3\) plays. Thus, no strategy profile lies in the core.
    \end{example}
\end{proof}

\subsection{Decision Problems}
We now turn our attention to decision problems relating to the core in the mean-payoff setting. However, from a computational perspective, there is an immediate concern here---given a potential beneficial deviation, how can we verify that it is preferable to the status quo under \emph{all} possible counter-responses? Fortunately, as we show in the following lemma, we can restrict our attention to \emph{memoryless strategies} when thinking about potential counter-responses to players' deviations (see Section~\ref{secn:prelim} for the definition of memoryless strategies):

\begin{lemma}\label{lemma:core-memoryless-counter-responses}
    Let \(G\) be a game, \(C \subseteq \Ag\) be a coalition and \(\vec{\sigma}\) be a strategy profile. Further suppose that \(\vec{\sigma}_C'\) is a strategy vector such that for all memoryless strategy vectors \(\vec{\sigma}_{\Ag \setminus C}'\), we have,
    \begin{equation*}
        \rho(\vec{\sigma}_C', \vec{\sigma}_{\Ag \setminus C}') \succ_i \rho(\vec{\sigma}).
    \end{equation*}
    Then, for all strategy vectors, \(\vec{\sigma}_{\Ag \setminus C}'\), not necessarily memoryless, we have,
    \begin{equation*}
        \rho(\vec{\sigma}_C', \vec{\sigma}_{\Ag \setminus C}') \succ_i \rho(\vec{\sigma}).
    \end{equation*}
\end{lemma}

Before we prove this, we need to introduce an auxiliary concept of \emph{two-player, turn-based, zero-sum, multi-mean-payoff games}~\cite{Velner2015} (we will simply refer to these as multi-mean-payoff games moving forward). Informally, these are similar to two-player, turn-based, zero-sum mean-payoff games, except player 1 has \(k\) weight functions associated with the edges, and they are trying to ensure the resulting \(k\)-vector of mean-payoffs is component-wise greater than a vector threshold. Formally, a multi-mean-payoff game is given by a structure \[G = (V_1, V_2, v^0, E, w, z^k)\] where \(V_1, V_2\) are sets of states controlled by players 1 and 2 respectively, with \(V = V_1 \cup V_2\) the state space, \(v^0 \in V\) the start state, \(E \subseteq V \times V\) a set of edges, \(w: E \to \mathbb{Z}^k\) a weight function, assigning to each edge a vector of weights, and \(z^k \in \mathbb{Q}^k\) is a threshold vector.

The game is played by starting in the start state, \(v^0 \in V_i\), and player \(i\) choosing an edge \((v^0, v^1)\), and traversing it to the next state. From this new state, \(v^1 \in S_j\), player \(j\) chooses an edge and so on, repeating this process forever. Runs are defined in the usual way and the payoff of a run \(\rho\), \(\pay(\rho)\), is simply the vector \((\mp(w_1(\rho)), \ldots,\mp(w_k(\rho)))\). Player 1 wins if the \(\pay_i(\rho) \geq z_i\) for all \(i \in \{1,\ldots,k\}\), and loses otherwise. The basic question associated with these games is whether player 1 can force a win:
\begin{quote}
    \underline{\textsc{Multi-Mean-Payoff-Threshold}}:\\
    \emph{Given}: Multi-mean-payoff game \(G\).\\
    \emph{Question}: Is it the case that player 1 has a winning strategy?
\end{quote}
As shown in~\cite{Velner2015}, this problem is co-NP-complete. Whilst we do not need to use this complexity result right now, we shall use this fact later. It is also worth noting that in our multi-player mean-payoff games, the weights are attached to states, whilst in multi-mean-payoff games, the weights are attached to edges. For our purposes, this difference is purely superficial---the former can be mapped into the latter simply by pushing the weights onto the outgoing edges, whilst the latter can be mapped into the former by adding more states, one for each edge. If you do this in the correct way, this gives you only a polynomial overhead. As we proceed, we shall use this mapping implicitly (and we in fact have already used it in the proof of Theorem~\ref{thm:empty-core}).

With this decision problem introduced, we are now in a position to prove Lemma~\ref{lemma:core-memoryless-counter-responses}.

\begin{proof}[Proof of Lemma~\ref{lemma:core-memoryless-counter-responses}]
    Let \(\vec{\sigma}_{\Ag \setminus C}^\prime\) be an arbitrary strategy and let \(i \in C\) be an arbitrary agent. Suppose it is not the case that \(\rho(\vec{\sigma}_C', \vec{\sigma}_{\Ag \setminus C}') \succ_i \rho(\vec{\sigma})\). Thus, we have \(\rho(\vec{\sigma}) \succeq_i \rho(\vec{\sigma}_C', \vec{\sigma}_{\Ag \setminus C}')\). Considering this as a two-player multi-mean-payoff game, where player 1's strategy is fixed and encoded into the game structure (\emph{i.e.}, player 1 follows \(\vec{\sigma}_C'\), but has no say in the matter), and the payoff threshold is \(\mp(\rho(\vec{\sigma}))\), then \(\vec{\sigma}_{\Ag \setminus C}'\) is a winning strategy for player~2 in this game. Now, by~\cite{Velner2015,Kopczynski2006}, if player~2 has a winning strategy,  then they have a memoryless winning strategy. Thus, there is a memoryless strategy \(\vec{\sigma}_{\Ag \setminus C}^{\prime\prime}\) such that \(\rho(\vec{\sigma}) \succeq_i \rho(\vec{\sigma}_C', \vec{\sigma}_{\Ag \setminus C}^{\prime\prime})\). But this contradicts the assumptions of the lemma, and thus we must have \(\rho(\vec{\sigma}_C', \vec{\sigma}_{\Ag \setminus C}') \succ_i \rho(\vec{\sigma})\).
\end{proof}

We are now in a position to look at some complexity bounds for mean-payoff games in the cooperative setting. Let us begin by considering the following decision problem relating to beneficial deviations:

\begin{quote}
    \underline{\textsc{Beneficial Deviation}}:\\
    \emph{Given}: Game \(G\) and strategy profile \(\vec{\sigma}\). \\
    \emph{Question}: Does some coalition have a beneficial deviation from \(\vec\sigma\)? That is, does there exist \(C\subseteq\Ag\) and \(\vec{\sigma}_C' \in \Sigma_C\) such that for all \(\vec{\sigma}_{\Ag \setminus C}' \in\Sigma_{\Ag \setminus C}\) and for all \(i \in C\), we have:
    \begin{equation*}
        \rho(\vec{\sigma}_C', \vec{\sigma}_{\Ag \setminus C}') \succ_i \rho(\vec{\sigma})?
    \end{equation*}
\end{quote}

We have:

\begin{theorem}\label{thm:mp-bdeviation}
    If the provided strategy profile \(\vec{\sigma}\) is memoryless, then \textsc{Beneficial Deviation} is NP-complete.
\end{theorem}
\begin{proof}
    First correctly guess a deviating coalition \(C\) and a strategy profile \(\vec{\sigma}'_C\) for such a coalition of players. Then, use the following three-step algorithm. First, compute the mean-payoffs that players in \(C\) get on \(\rho(\vec{\sigma})\), that is, a set of values \(z^*_j = \pay_j(\rho(\vec{\sigma}))\) for every \(j\in C\) --- this can be done in polynomial time simply by `running' the strategy profile~\(\vec{\sigma}\). Then compute the graph \(G[\vec{\sigma}'_C]\), which contains all possible behaviours ({\em i.e.}, strategy profiles) for \(\Ag \setminus C\) with respect to \(\vec{\sigma}\) --- this construction is similar to the one used in the proof of Theorem~\ref{thm:bdeviation}, that is, the game when we fix \(\vec{\sigma}'_C\), and can be done in polynomial time. Finally, we ask whether every path \(\rho\) in \(G[\vec{\sigma}'_{C}]\) satisfies \(\pay_j(\rho) > z^*_j\), for every \(j\in C\) --- for this step, we can use Karp's algorithm~\cite{Karp1978} to answer the question in polynomial time for every \(j\in C\). If every path in \(G[\vec{\sigma}'_{C}]\) has this property, then we accept; otherwise, we reject.

    For hardness, we reduce from 3SAT, using a small variation of the construction in~\cite{SistlaC85}. Let \(P=\{x_1,\ldots,x_n\}\) be a set of atomic propositions. Given a Boolean formula \(\varphi = \bigwedge_{1\leq c \leq m} C_c\) (in conjunctive normal form) over \(P\) --- where each \(C_c = l_{c1} \vee l_{c2} \vee l_{c3}\), and each literal \(l_{ck}=x_j\) or \(\neg x_j\), with \(1\leq k\leq 3\), for some \(1\leq j \leq n\) --- we construct \(M = \left(\Ag, \St, s^0, {(\Ac_i)}_{i \in \Ag}, \transf\right)\), an \(m\)-player concurrent game structure defined as follows, and illustrated in Figure~\ref{fig:coop-bd-np-hard}:
    \begin{itemize}
        \item \(\Ag = \{1,\ldots,m\}\);
        \item \(\St = \{x_v \ | \ 1\leq v\leq n\} \cup \{x'_v \ | \ 1\leq v\leq n\} \cup \{y_0,y_n,y^0,y^*\}\);
        \item \(s^0 = y^0\);
        \item \(\Ac_{i} = \{t,f\}\), for every \(i\in\Ag\), and \(\Ac = \Ac_1 \times \cdots \times \Ac_m\);
        \item For \(\transf\), refer to the Figure~\ref{fig:coop-bd-np-hard}, such that \(T=\{(t_1,\ldots,t_m)\}\) and \(F=\Ac\setminus T\).
    \end{itemize}

    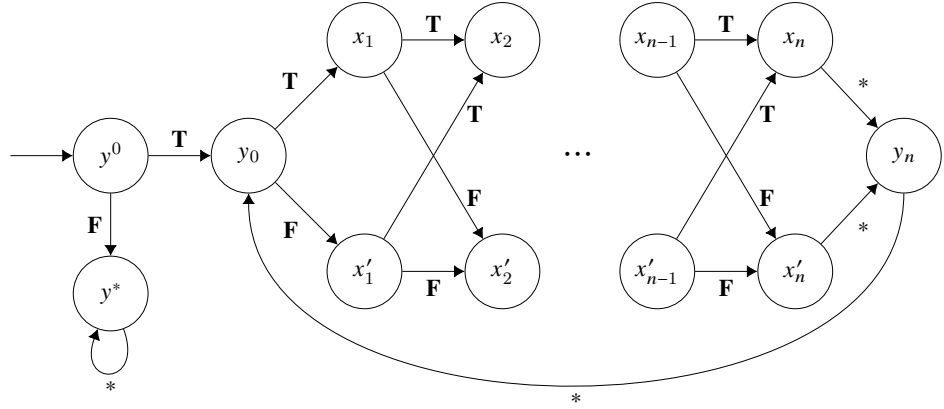
\begin{figure}
        \centering
        \resizebox{\columnwidth}{!}{
            \begin{tikzpicture}[state/.style={circle, draw, minimum size = 1.2cm}, invisible/.style={}]
				\node[state] (y^0) {$y^0$};
				\node[invisible, inner sep=0pt] (start) [left = of y^0] {};
				\node[state] (y_0) [right = of y^0] {$y_0$};
				\node[state] (y^*) [below= of y^0] {$y^*$};		
				\node[state] (x_1) [above right = of y_0] {$x_1$};
				\node[state] (x_1^prime) [below right = of y_0] {$x_1^\prime$};
				\node[state] (x_2) [right = of x_1] {$x_2$};
				\node[state] (x_2^prime) [right = of x_1^prime] {$x_2^\prime$};
				\node[invisible] (cdots_above) [right = 0.5cm of x_2] {};
				\node[invisible] (cdots_below) [right = 0.5cm of x_2^prime] {};
				\node[invisible] (cdots) at ($0.5*(cdots_above)+0.5*(cdots_below)$) {$\bm{\cdots}$};
				\node[state] (x_{n-1}) [right = 0.5cm of cdots_above] {$x_{n-1}$};
				\node[state] (x_{n-1}^prime) [right = 0.5cm of cdots_below] {$x_{n-1}^\prime$};
				\node[state] (x_n) [right = of x_{n-1}] {$x_n$};
				\node[state] (x_n^prime) [right = of x_{n-1}^prime] {$x_n^\prime$};
				\node[invisible] (y_n_above) [right = of x_n] {};
				\node[invisible] (y_n_below) [right = of x_n^prime] {};
				\node[state] (y_n) at ($0.5*(y_n_above)+0.5*(y_n_below)$) {$y_n$};

				\draw [-{Latex[width=2mm]}]
				(start) edge node{} (y^0)
				(y^0) edge node[left]{$\FALSE$} (y^*)
				(y^*) edge[loop, out=290, in=250, distance=1cm] node[below]{$*$} (y^*)
				(y^0) edge node[above]{$\TRUE$} (y_0)

				(y_0) edge node[above left]{$\TRUE$} (x_1)
				(y_0) edge node[below left]{$\FALSE$} (x_1^prime)

				(x_1) edge node[above]{$\TRUE$} (x_2)
				(x_1) edge node[right,pos=0.75]{$\FALSE$} (x_2^prime)
				(x_1^prime) edge node[right,pos=0.75]{$\TRUE$} (x_2)
				(x_1^prime) edge node[below]{$\FALSE$} (x_2^prime)

				(x_{n-1}) edge node[above]{$\TRUE$} (x_n)
				(x_{n-1}) edge node[right,pos=0.75]{$\FALSE$} (x_n^prime)
				(x_{n-1}^prime) edge node[right,pos=0.75]{$\TRUE$} (x_n)
				(x_{n-1}^prime) edge node[below]{$\FALSE$} (x_n^prime)

				(x_n) edge node[above right]{$*$} (y_n)
				(x_n^prime) edge node[below right]{$*$} (y_n)

				(y_n) edge[out=270, in=270] node[below]{$*$} (y_0)
				;
\end{tikzpicture}
        }
        \caption{Concurrent game structure for the reduction from 3SAT.\label{fig:coop-bd-np-hard}}
    \end{figure}

    With \(M\) at hand, we build a mean-payoff game using the following weight function:
    \begin{itemize}
        \item \(w_i(x_v) = 1\) if \(x_v\) is a literal in \(C_i\) and \(w_i(x_v) = 0\) otherwise, for all \(i\in\Ag\) and \(1\leq v\leq n\)
        \item \(w_i(x'_v) = 1\) if \(\neg x_v\) is a literal in \(C_i\) and \(w_i(x'_v) = 0\) otherwise, for all \(i\in\Ag\) and \(1\leq v\leq n\)
        \item \(w_i(y_0)=w_i(y_n)=w_i(y^0)=w_i(y^*)=0\), for all \(i\in\Ag\)
    \end{itemize}

    Then, we consider the game \(G\) over \(M\) and any strategy profile (in memoryless strategies) such that \(\vec{\sigma}(s^0)=y^*\). For any of such strategy profiles the mean-payoff of every player is 0. However, if \(\varphi\) is satisfiable, then there is a path in \(M\), from \(y_0\) to \(y_n\), such that in such a path, for every player, there is a state in which its payoff is not 0. Thus, the grand coalition \(\Ag\) has an incentive to deviate since traversing that path infinitely often will give each player a mean-payoff strictly greater than 0. Observe two things. Firstly, that only if the grand coalition \(\Ag\) agrees, the game can visit \(y_0\) after \(y^0\). Otherwise, the game will necessarily end up in \(y^*\) forever after. Secondly, because we are considering memoryless strategies, the path from \(y_0\) to \(y_n\) followed at the beginning is the same path that will be followed thereafter, infinitely often. Then, we can conclude that there is a beneficial deviation (necessarily for \(\Ag\)) if and only if \(\varphi\) is satisfiable, as otherwise at least one of the players in the game will not have an incentive to deviate (because its mean-payoff would continue to be 0). We then conclude that \((G,\sigma)\in\text{\textsc{Beneficial Deviation}}\) if and only if \(\varphi\) is satisfiable.
\end{proof}

From Theorem~\ref{thm:mp-bdeviation} it follows that checking if no coalition of players has a beneficial deviation with respect to a given strategy profile is co-NP complete, and thus we have:

\begin{corollary}\label{prop:coop-mem-complete-mless}
    If the provided strategy profile \(\vec{\sigma}\) is memoryless, then \textsc{Core Membership} is co-NP-complete.
\end{corollary}

We can also leverage \textsc{Beneficial Deviation} to obtain the following:

\begin{theorem}\label{prop:coop-ecore-mless}
    If we restrict the existential quantification over strategies to consider only memoryless strategies, then \textsc{E-Core} is in \(\Sigma^P_2\).
\end{theorem}
\begin{proof}
    Given a game \(G\), we guess a strategy profile \(\vec{\sigma}\) and check that \((G,\vec{\sigma})\) is not an instance of \textsc{Beneficial Deviation}. While the former can be done in polynomial time, the latter can be solved in co-NP using an oracle for \textsc{Beneficial Deviation}. Thus, we have a procedure that runs in NP\,\(^{\text{co-NP}}\) = \(\Sigma^P_2\).
\end{proof}

Owing to the alternations present in the definition of the core and beneficial deviations, we believe the \textsc{E-Core} problem to be complete for \(\Sigma^P_2\), but we have been unable to find a proof of this. Note that suggests a contrast with the corresponding problem for Nash equilibrium in mean-payoff games, which lies in NP~\cite{Ummels2011}. More importantly, the result also shows that the (complexity) dependence on the type of coalitional deviation is only weak, in the sense that different types of beneficial deviations may be considered within the same complexity class, as long as such deviations can be checked with an NP or co-NP oracle.

\subsection{Fulfilled Coalitions}
We now generalise the idea of a fulfilled coalition. To do this, we introduce the notion of a \textbf{lower bound}. Let \(C\subseteq\Ag\) be a coalition in a game \(G\) and let \(\vec{z}_{C}\in \mathbb{Q}^C\). We say that \(\vec{z}_C\) is a {lower bound} for \(C\) if there is a joint strategy \(\vec{\sigma}_C\) for \(C\) such that for all strategies \(\vec{\sigma}_{-C}\) for \(\Ag\setminus C\), we have \(\pay_i(\rho(\vec{\sigma}_C,\vec{\sigma}_{-C}))\geq z_i\), for every \(i\in C\). Based on this definition, we can prove the following, which characterises the core in terms of lower bounds:

\begin{lemma}\label{lmma:eq_to_path_finding_core}
    Let \(\rho\) be a run in \(G\). There is \(\vec{\sigma} \in \coreset{G}\) such that \(\rho = \rho(\vec{\sigma})\) if and only if for every coalition \(C\subseteq\Ag\) and lower bound \(\vec{z}_C \in \mathbb{Q}^C\) for \(C\), there is some \(i\in C\) such that \(z_i \leq \pay_i(\rho)\).
\end{lemma}

\begin{proof}
    To show the left-to-right direction, suppose that there exists a member of the core \(\vec{\sigma} \in \coreset{G}\)  and suppose further that there is some coalition \(C\subseteq\Ag\) and lower bound \(\vec{z}_C \in \mathbb{Q}^C\) for \(C\), such that for every \(i\in C\) we have \(z_i > \pay_i(\rho)\). Because \(\vec{z}_C\) is a lower bound for \(C\), and \(z_i > \pay_i(\rho)\), for every \(i\in C\), then there is a joint strategy \(\vec{\sigma}_C\) for \(C\) such that for all strategies \(\vec{\sigma}_{-C}\) for \(\Ag\setminus C\), we have \(\pay_i(\rho(\vec{\sigma}_C,\vec{\sigma}_{-C}))\geq z_i > \pay_i(\rho)\), for every \(i\in C\). Then, it follows that \((G,\vec{\sigma})\in\text{\textsc{Beneficial Deviation}}\), which further implies that \(\vec{\sigma}\) cannot be in the core of \(G\) --- a contradiction to our initial hypothesis.

    For the right-to-left direction, suppose that there is \(\rho\) in \(G\) such that for every coalition \(C\subseteq\Ag\) and lower bound \(\vec{z}_C \in \mathbb{Q}^C\) for \(C\), there is \(i\in C\) such that \(z_i \leq \pay_i(\rho)\). We then simply let \(\vec{\sigma}\) be any strategy profile such that \(\rho = \rho(\vec{\sigma})\). Now, let \(C=\{j,\ldots,k\}\subseteq\Ag\) be any coalition and \(\vec{\sigma}'_{C}\) be any possible deviation of \(C\) from \(\vec{\sigma}\). Either \(\vec{z'}_C = (\pay_j(\rho(\vec{\sigma}_{-C},\vec{\sigma}'_C)),\ldots,\pay_k(\rho(\vec{\sigma}_{-C},\vec{\sigma}'_C)))\) is a lower bound for \(C\) or it is not.

    If we have the former, by hypothesis, we know that there is \(i\in C\) such that \(\pay_i(\rho(\vec{\sigma}_{-C},\vec{\sigma}'_C)) \leq \pay_i(\rho)\). Therefore, \(i\) will not have an incentive to deviate along with \(C\setminus\{i\}\) from \(\vec{\sigma}\), and as a consequence coalition \(C\) will not be able to beneficially deviate from \(\vec{\sigma}\).

    If, on the other hand, \(\vec{z'}_C\) is not a lower bound for \(C\), then, by the definition of lower bounds, we know that it is not the case that \(\vec{\sigma}'_C\) is a joint strategy for \(C\) such that for all strategies \(\vec{\sigma}'_{-C}\) for \(\Ag\setminus C\), we have \(\pay_i(\rho(\vec{\sigma}'_C,\vec{\sigma}'_{-C}))\geq \pay_i(\rho(\vec{\sigma}_{-C},\vec{\sigma}'_C))\), for every \(i\in C\). That is, there exists \(i\in C\) and \(\vec{\sigma}'_{-C}\) for \(\Ag\setminus C\) such that \(\pay_i(\rho(\vec{\sigma}'_C,\vec{\sigma}'_{-C})) < \pay_i(\rho(\vec{\sigma}_{-C},\vec{\sigma}'_C))\). We will now choose \(\vec{\sigma}'_{-C}\) so that, in addition, \(\pay_i(\rho) \geq \pay_i(\rho(\vec{\sigma}'_C,\vec{\sigma}'_{-C}))\) for some \(i\).

    Let \(\vec{z''}_C = (\pay_j(\rho(\vec{\sigma}^j_{-C},\vec{\sigma}'_C)),\ldots,\pay_k(\rho(\vec{\sigma}^k_{-C},\vec{\sigma}'_C)))\) where \(\pay_i(\rho(\vec{\sigma}^i_{-C},\vec{\sigma}'_C))\) is defined to be \(\min_{\vec{\sigma}'_{-C}\in\Sigma_{-C}} \pay_i(\rho((\vec{\sigma}'_{-C}, \vec{\sigma}'_C)))\). That is, \(\vec{\sigma}^i_{-C}\) is a strategy for \(\Ag\setminus C\) which ensures the lowest mean-payoff for \(i\) assuming that \(C\) is playing the joint strategy \(\vec{\sigma}'_C\). By construction \(\vec{z''}_C\) is a lower bound for \(C\) --- since each \(z''_i = \pay_i(\rho(\vec{\sigma}^i_{-C},\vec{\sigma}'_C))\) is the greatest mean-payoff value that \(i\) can ensure for itself when \(C\) is playing \(\vec{\sigma}'_C\), no matter what coalition \(\Ag\setminus C\) does---and therefore, by hypothesis we know that for some \(i\in C\) we have \(\pay_i(\rho(\vec{\sigma}^i_{-C},\vec{\sigma}'_C)) \leq \pay_i(\rho)\). As a consequence, as before, \(i\) will not have an incentive to deviate along with \(C\setminus\{i\}\) from \(\vec{\sigma}\), and therefore coalition \(C\) will not be able to beneficially deviate from \(\vec{\sigma}\). Because \(C\) and \(\vec{\sigma}'_C\) where arbitrarily chosen, we conclude that \(\vec{\sigma}\in\coreset{G}\), proving the right-to-left direction and finishing the proof.
\end{proof}

With this lemma in mind, we want to determine if a given vector, \(\vec{z}_C\), is in fact a lower bound and importantly, how efficiently we can do this. That is, to understand the following decision problem:

\begin{quote}
    \underline{\textsc{Lower Bound}}:\\
    \emph{Given}: Game \(G\), coalition \(C \subseteq \Ag\), and vector \(\vec{z}_C \in \mathbb{Q}^{\Ag}\). \\
    \emph{Question}: Is \(\vec{z}_C\) a lower bound for \(C\) in \(G\)?
\end{quote}

Using the \textsc{Multi-Mean-Payoff-Threshold} decision problem introduced earlier, we can prove the following theorem:

\begin{theorem}
    \textsc{Lower Bound} is co-NP-complete.
\end{theorem}

\begin{proof}
    We prove membership as well as hardness by reducing to and from \textsc{Multi-Mean-Payoff-Threshold} in the obvious way. First, we show that \textsc{Lower Bound} lies in co-NP by reducing it to \textsc{Multi-Mean-Payoff-Threshold}. Suppose we have an instance, \((G, C, \vec{z}_C)\), and we want to determine if it is in \textsc{Lower-Bound}. We can do this by forming a two-player, multi-mean-payoff game, \(G' = (V_1, V_2, v^0, E, w', z^k)\). Here we have \(V_1 = \St\), \(V_2 = \St \times \Ac_C\) and \(v^0 = s^0\). Additionally, the set of edges of \(G'\), \(E\), is defined as,
    \begin{align*}
        E & =   \{(s, (s, \ac_C)) \mid s \in \St, \ac_C \in \Ac_C\}                                                                  \\
          & \cup  \{((s, \ac_C), \transf(s, (\ac_C, \ac_{\Ag \setminus C}))) \mid \ac_{\Ag \setminus C} \in \Ac_{\Ag \setminus C}\},
    \end{align*}
    and the weight function, \(w' : E \to \mathbb{Z}^{\abs{C}}\), is defined by the following two patterns:
    \begin{align*}
         & w_i'(s, (s, \ac_C)) =w_i(s);                                           \\
         & w_i'((s, \ac_C), \transf(s, (\ac_C, \ac_{\Ag \setminus C}))) = w_i(s).
    \end{align*}
    Finally, we set \(z^{\abs{C}}\) to be \(\vec{z}_c\).

    Informally, the two players of the game are \(C\) and \(\Ag \setminus C\), the vector weight function is given by aggregating the weight functions of \(C\) and the threshold is \(\vec{z}_C\). Now, if in this game, player 1 has a winning strategy, then there exists some strategy \(\vec{\sigma}_C\) such that for all strategies of player 2, \(\vec{\sigma}_{\Ag \setminus C}\), we have that \(\rho(\vec{\sigma}_C, \vec{\sigma}_{\Ag \setminus C})\) is a winning run for player 1. But this means that \(\pay_i(\rho(\vec{\sigma}_C, \vec{\sigma}_{\Ag \setminus C})) \geq z_i\) for all \(i \in C\). But it is easy to verify that this implies that \(\vec{z}_C\) is a lower bound for \(C\) in \(G\). Conversely, if player 1 has no winning strategy, then for all strategies, \(\vec{\sigma}_C\), there exists some strategy \(\vec{\sigma}_{\Ag \setminus C}\) such that \(\rho(\vec{\sigma}_C, \vec{\sigma}_{\Ag \setminus C})\) is not a winning run. This in turn implies that for some \(j \in C\), we have that \(\pay_j(\rho(\vec{\sigma}_C, \vec{\sigma}_{\Ag \setminus C})) < z_j\), which means that \(\vec{z}_C\) is not a lower bound for \(C\) in \(G\). Also note that this construction can be performed in polynomial time, giving us the co-NP upper bound.

    For the lower bound, we go the other way and reduce from \textsc{Multi-Mean-Payoff-Threshold}. Suppose we would like to determine if an instance \(G\) is in \textsc{Multi-Mean-Payoff-Threshold}. Then we form a concurrent mean-payoff game, \(G'\), with \(k+1\) players, where the states of \(G'\) coincide exactly with the states of \(G\). In this game, only the \(1^{\text{st}}\) and \({(k+1)}^{\text{th}}\) player have any influence on the strategic nature of the game. If the game is in a state in \(V_1\), player one can decide which state to move into next. Otherwise, if the game is in a state within \(V_2\), then the \({(k+1)}^{\text{th}}\) player makes a move. Note we only allow moves that agree with moves allowed within \(G\).

    Now, in \(G'\), the first \(k\) players have weight functions corresponding correspond to the \(k\) weight functions of player 1 in \(G\). The last player can have any arbitrary weight function. With this machinery in place, we ask if \(z^k\) is a lower bound for \(\{1,\ldots, k\}\). In a similar manner of reasoning to the above, it is easy to verify that \(G\) is an instance of \textsc{Multi-Mean-Payoff-Threshold} if and only if \(z^k\) is a lower bound for \(\{1,\ldots, k\}\) in the constructed concurrent mean-payoff game. Moreover, this reduction can be done in polynomial time and we can conclude that \textsc{Lower Bound} is co-NP-complete.
\end{proof}

A notable omission from this section is that we have not presented any bounds for the complexity of \textsc{E-Core} in the general case. One reason for the upper bounds remaining elusive to us is due to the fact that whilst in a multi-mean-payoff game, player 2 can act optimally with memoryless strategies, player 1 may require infinite memory~\cite{Velner2015, Kopczynski2006}. Given the close connection between the core in our concurrent, multi-agent setting and winning strategies in multi-mean-payoff games, this raises computational concerns for the \textsc{E-Core} problem. Additionally, in~\cite{Brenguier2015}, the authors study the Pareto frontier of multi-mean-payoff games, and provide an algorithm for constructing a representation of the achievable values of a given game, but this procedure takes an exponential amount of time in the size of the input. The same paper also establishes \(\Sigma^p_2\)-completeness for the \emph{polyhedron value problem}. Both of these problems appear to be intimately related to the core, and we hope we might be able to use these results to gain more insight into the \textsc{E-Core} in the future.

\section{Discussion and Related Work}\label{secn:conc}

In this section, we present some conclusions, briefly discuss related work, discuss some issues related to the \core, and speculate about how to implement our approach using model checking techniques.

\subsection{Coalition formation in cooperative games}

Coalition formation with externalities has been studied in the cooperative game-theory literature~\cite{yi1997stable,uyanik2015nonemptiness,Finus2003}. These works considered several possible formulations of the \core. For instance, the \(\alpha\)-\core\ takes the pessimistic approach that requires that all members of a deviating coalition, \(S\), will benefit from the deviation regardless of the behaviour of the other coalitions that may be formed. Our first definition of the \core\ follows this approach. In contrast, \(\beta\)-\core\ takes an optimistic approach, and requires that the members of a deviating coalition \(S\) will benefit from at least one possibility of coalition formation of the rest of the players. In addition, \(\gamma\)-\core~\cite{Chander2010,Chander2007} assumes that the coalition structure that will be created after a deviation will include the deviating coalition \(S\) and the rest of the coalition structure will consist of all singletons. The ``worth'' of \(S\) is now defined as equal to its payoff in the Nash equilibrium between \(S\) and the other players acting individually, in which the members of \(S\) play their joint best response strategy against the individually best response strategies of the remaining players. It is well-known that \(\alpha\)- and \(\beta\)-characteristic functions lead to large cores~\cite{ray1997equilibrium}, which is consistent with our observation that, with respect to our first definition, the \core\ is never empty. Coalition formation is important in multi-agent system~\cite{shehory1998methods}. However, even though coalition formation with externalities is a widely-studied problem in multi-agent systems, not much work has studied the concept of stability in multi-agent coalition formation with externalities~\cite{mutzari2021coalition}. Instead, in artificial intelligence and multi-agent systems, most research has focused on the structure formation itself~\cite{rahwan2012anytime}. Notice that, from the point of view of cooperative game theory, our games are games of non-transferable utility (i.e., NTU games)~\cite[p.71]{Chalkiadakis2011}. For certain types of NTU games, there are other possible approaches to defining the core, for example through a translation into a conventional (transferable utility) game (see, e.g.~\cite[p.238]{Borm2015}). This approach might conceivably be used in our setting.

\subsection{Rational verification of concurrent games}

The formal verification of temporal logic properties of multi-agent systems, while assuming rational behaviour of the agents in such a system, has been studied for almost a decade now; see, for instance,~\cite{Gutierrez2015,Gutierrez2017a,WooldridgeGHMPT16,Kupferman2016,Kupferman2007}. However, to the best of our knowledge, all these studies have considered a non-cooperative setting, even if coalitional power is allowed, for instance, as in a strong Nash equilibrium. Nonetheless, also in such non-cooperative settings, the complexity of checking whether a temporal logic property is satisfied in a stable outcome of the game is a \textsc{2ExpTime}-complete problem, even for two-player zero-sum games where only trivial coalitions can be formed. On the positive side, cooperative games seem to have better model-theoretic properties in the rational verification framework: with respect to our first definition of the \core\ (which corresponds to the concept of \(\alpha\)-\core\ in the literature of cooperative games), a witness in the core is always guaranteed (since the \core\ is never empty), preserved across bisimilar systems, and can be checked in practice using \atlstar\ model checking techniques, which are supported by, \emph{e.g.}, MCMAS~\cite{Lomuscio2017}, an automated formal verification tool that can perform \atlstar\ model checking (via an implementation of SL[1G], which subsumes \atlstar~\cite{CermakLM15}) and allows specifications of concurrent game structures in ISPL (Interpreted Systems Programming Language, the modelling language of MCMAS, based on the interpreted systems formalism~\cite{Fagin1995}).


\paragraph*{\bf Acknowledgment}
Sarit Kraus was partly supported by the by the Israel Science Foundation (grants No. 
1958/20).
Thomas Steeples gratefully acknowledges the support of the EPSRC Centre for Doctoral Training in Autonomous Intelligent Machines and Systems EP/L015897/1, along with the Ian Palmer Memorial Scholarship.
Michael Wooldridge was supported by a UKRI Turing AI World Leading Researcher Fellowship (grant EP/W002949/1).
Both Sarit Kraus and Michael Wooldridge were also supported by the EU project TAILOR (Grant 992215).
\bibliography{references}

\appendix 

\section{Proof that Example~\ref{ex:empty-strong-core} has an empty strong core}

We claimed earlier that there exist four players games with an empty strong core, and gave an example of such a game in Example~\ref{ex:empty-strong-core}. Below, we present a table giving, for every action profile, a strong beneficial deviation for some set of players. The notation \(\{2 \mapsto 0, 4 \mapsto 1\}\) means that the coalition \(\{2,4\}\) have a strong beneficial deviation where player 2 plays \(0\) and player 4 player \(1\). Note that a strong beneficial deviation for a coalition does not require all players to change their action from the status quo.

\begin{figure}[H]
\begin{tabular}{llll}
\toprule
Action profile & Set of winning players & Strong beneficial deviation \\
\midrule
\(0000\) & \(\{1,2\}\) & \(\{3 \mapsto 1 \}\) \\
\(0001\) & \(\{1,2\}\) & \(\{3 \mapsto 1\}\) \\
\(0010\) & \(\{1,3\}\) & \(\{2 \mapsto 1, 4 \mapsto 1\}\) \\ 
\(0011\) & \(\{1,3\}\) & \(\{2 \mapsto 1, 4 \mapsto 1\}\) \\
\(0100\) & \(\{1,2\}\) & \(\{3 \mapsto 1\}\) \\
\(0101\) & \(\{1,4\}\) & \(\{2 \mapsto 0\}\) \\
\(0110\) & \(\{1,3\}\) & \(\{2 \mapsto 1 ,4 \mapsto 1\}\) \\
\(0111\) & \(\{2,4\}\) & \(\{1 \mapsto 1\}\) \\
\(1000\) & \(\{3,4\}\) & \(\{1 \mapsto 0\}\) \\
\(1001\) & \(\{2,3\}\) & \(\{1 \mapsto 0\}\) \\
\(1010\) & \(\{2,3\}\) & \(\{1 \mapsto 0\}\) \\
\(1011\) & \(\{2,3\}\) & \(\{1 \mapsto 0\}\) \\
\(1100\) & \(\{1,2\}\) & \(\{3 \mapsto 1\}\) \\
\(1101\) & \(\{1,4\}\) & \(\{2 \mapsto 0\}\) \\
\(1110\) & \(\{2,3\}\) & \(\{1 \mapsto 0\}\) \\
\(1111\) & \(\{1,4\}\) & \(\{2 \mapsto 0\}\) \\
\bottomrule
\end{tabular}
\end{figure}

This table illustrates that the provided game has an empty strong core. It is an easy, but somewhat tedious, task to verify that the deviating players win under all counter-responses under which the original winning players win.

\section{Proof of Theorem~\ref{thm:strong-core-nonempty}}

We will prove this theorem by assuming the existence of a three player game \(G\) with an empty strong core, and then proving two contradictory facts, which immediately implies the result. The two claims are:
\begin{enumerate}
\item There is some strategy \(\vec{\sigma}\) on \(G\) that models exactly two of \(\gamma_1\), \(\gamma_2\), \(\gamma_3\).
\item There is no strategy \(\vec{\sigma}\) on \(G\) that models exactly two of \(\gamma_1\), \(\gamma_2\), \(\gamma_3\).
\end{enumerate}

We use the following lemma throughout:

\begin{lemma}
  Let \(G\) be a three-player game with an empty strong core. Then there does not exist a strategy profile which satisfies the goals of all three players.
\end{lemma}

\begin{proof}
If \(G\) had a strategy satisfying all of the players' goals, such a strategy would be in the strong core, which is a contradiction. 
\end{proof}

The first claim we need to prove is the more straightforward of the two:

\begin{lemma}\label{lmm:strong-core-three-player-non-empty-first-part}
        Let \(G\) be a three-player game with an empty strong core. Then there exists some strategy \(\vec{\sigma}\) on \(G\) that models exactly two of \(\gamma_1\), \(\gamma_2\), \(\gamma_3\).
\end{lemma}

\begin{proof}
Given claim 1, we know that every strategy on \(G\) must model exactly zero, one or two goals, but not three. Suppose, for contradiction, that \(G\) only has strategies modelling none or one of \(\gamma_1\), \(\gamma_2\), \(\gamma_3\). First, note that if \(G\) only had strategies satisfying zero of \(\gamma_1\), \(\gamma_2\), \(\gamma_3\), none of those strategies could admit a strong beneficial deviation, and so would all be in the \ccore. Hence, \(G\) must have a strategy satisfying exactly one of \(\gamma_1\), \(\gamma_2\), \(\gamma_3\). Let such a strategy be denoted \(\vec{\sigma} = (\sigma_1, \sigma_2, \sigma_3)\), and without loss of generality, we may assume that it satisfies only \(\gamma_1\). Now, as by assumption \(\vec{\sigma}\) does not lie in \(\ccoreset{G}\), there must be a strong beneficial deviation from this strategy for some coalition \(C \subseteq \{2,3\}\). Since under a strong beneficial deviation, every member of the coalition is winning, and since we assumed that \(G\) only has strategies satisfying less than two of the goals, the coalition must be a singleton. Again, without loss of generality, we may assume that it is player 2 who has a deviation \(\sigma^\prime_2\). By assumption, we have that \(\sigma_1, \sigma_2^\prime, \sigma_3\) only models \(\gamma_2\).

But again, \((\sigma_1, \sigma^\prime_2, \sigma_3)\) must have a strong beneficial deviation for \(1\) or \(3\) to a strategy satisfying only \(\gamma_1\), or only \(\gamma_3\), respectively. If the former case were true, then this would contradict the fact that \(\sigma_2^\prime\) was a strong beneficial deviation from \(\vec{\sigma}\). Hence, we must be in the latter case, and have a strong beneficial deviation \(\sigma^\prime_3\) for player 3.

The new strategy, \((\sigma_1, \sigma^\prime_2, \sigma^\prime_3)\) must once again admit a strong beneficial deviation. If it admitted a strong beneficial deviation \(\sigma^\prime_1\) for \(1\) satisfying \(\gamma_1\) only, then the joint strategy \((\sigma^\prime_1, \sigma^\prime_3)\) would be a punishment for \(\sigma^\prime_2\). Otherwise, if a strong beneficial deviation \(\sigma''_2\) for 2 to a strategy satisfied \(\gamma_2\) only, then \((\sigma_1, \sigma''_2)\) would constitute a punishment for \(\sigma^\prime_3\). Either way, we reach a contradiction, concluding our proof.
\end{proof}

Before proving the second claim, we will need to appeal to two additional technical lemmas. The first is the following:

\begin{lemma}\label{lmm:strong-core-three-player-non-empty-technical-lemma-1}
Let \(\{a,b,c\}\) be a permutation of \(\{1,2,3\}\). Then there do not exist strategies \((\sigma_a, \sigma_b, \sigma_c)\), \((\sigma_a, \sigma_b, \sigma^\prime_c)\), \((\sigma^\prime_a, \sigma_b, \sigma^\prime_c)\) of \(G\) such that all of the following hold:
\begin{itemize}
\item{\((\sigma_a, \sigma_b, \sigma_c)\) satisfies \(\gamma_a\) and \(\gamma_b\) but not \(\gamma_c\)}
\item{\((\sigma_a, \sigma_b, \sigma^\prime_c)\) satisfies \(\gamma_c\) only}
\item{\((\sigma^\prime_a, \sigma_b, \sigma^\prime_c)\) satisfies \(\gamma_a\) only}
\item{\(\sigma^\prime_c\) is a beneficial deviation for \(c\) from \((\sigma_a, \sigma_b, \sigma_c)\) to \((\sigma_a, \sigma_b, \sigma^\prime_c)\)}
\item{\(\sigma^\prime_a\) is a beneficial deviation for \(a\) from \((\sigma_a, \sigma_b, \sigma^\prime_c)\) to \((\sigma^\prime_a, \sigma_b, \sigma^\prime_c)\)}
\end{itemize}
\end{lemma}

\begin{proof}
For a contradiction, we assume the statement of the lemma is not the case, and all of the statements hold. Without loss of generality, we can also assume that \(a=1\), \(b=2\), and \(c=3\).

Now, the third strategy, \((\sigma^\prime_1, \sigma_2, \sigma^\prime_3)\) must have a strong beneficial deviation for \(2\), \(3\), or both. If it has a deviation \((\sigma^\prime_2, \sigma''_3)\) (with possibly \(\sigma^\prime_2 = \sigma_2\)) to a strategy that satisfies \(\gamma_3\) but not \(\gamma_1\), then it would be a punishment for \(\sigma^\prime_1\). Therefore, the only possible sets of winners in the new strategy are \(\{2\}\), \(\{1,2\}\) and \(\{1,3\}\). We consider each of these cases in turn.

If player 2 is the only winner under the new strategy, the strong beneficial deviation must be \(\sigma^\prime_2\) by \(2\) only, to the strategy \((\sigma^\prime_1, \sigma^\prime_2, \sigma^\prime_3)\).

Again, this strategy has to have a beneficial deviation for \(1\), \(3\), or both. Let that deviation be \((\sigma''_1, \sigma''_3)\) with possibly \(\sigma''_1\) = \(\sigma^\prime_1\) or \(\sigma''_3\) = \(\sigma^\prime_3\) but not both. Consider the following cases for the resulting strategy, \((\sigma''_1, \sigma^\prime_2, \sigma''_3)\):
\begin{itemize}
\item If only \(\gamma_1\), or exactly \(\gamma_1\) and \(\gamma_3\) are satisfied, then \((\sigma''_1, \sigma''_3)\) is a punishment for \(\sigma^\prime_2\); a contradiction.
\item If exactly \(\gamma_1\) and \(\gamma_2\) are satisfied, then since \(3\) loses, we must necessarily have \(\sigma''_3 = \sigma^\prime_3\) (the deviation is just by \(1\)). In this case, the joint strategy \((\sigma''_1, \sigma^\prime_2)\) can take us from \((\sigma_1, \sigma_2, \sigma^\prime_3)\) to \((\sigma''_1, \sigma^\prime_2, \sigma''_3)\) and form a punishment for \(\sigma^\prime_3\); a contradiction.
\item If only \(\gamma_3\), or exactly \(\gamma_2\) and \(\gamma_3\) are satisfied, then since either way \(1\) loses, we must have \(\sigma''_1 = \sigma^\prime_1\) and the joint strategy \((\sigma^\prime_2, \sigma''_3)\) can take us from \((\sigma^\prime_1, \sigma_2, \sigma^\prime_3)\) to \((\sigma''_1, \sigma^\prime_2, \sigma''_3)\) and form a punishment for \(\sigma^\prime_1\); a contradiction.
\end{itemize}

Now suppose that the winners in the new strategy are \(1\) and \(2\). Since \(1\) was a winner in \((\sigma^\prime_1, \sigma_2, \sigma^\prime_3)\) already, the strong beneficial deviation must once again be \(\sigma^\prime_2\) by \(2\) only, to the strategy \((\sigma^\prime_1, \sigma^\prime_2, \sigma^\prime_3)\).

But now, we can see that the joint strategy \((\sigma^\prime_1, \sigma^\prime_2)\) is a punishment for \(\sigma^\prime_3\) (both \(1\) and \(2\) remain winners); a contradiction.

Finally, suppose the winners in the new strategy are \(1\) and \(3\). Since \(1\) was a winner in \((\sigma^\prime_1, \sigma_2, \sigma^\prime_3)\) already, the strong beneficial deviation must this time be \(\sigma''_3\) by \(3\) only, to the strategy \((\sigma^\prime_1, \sigma_2, \sigma''_3)\).

The only loser, \(2\), must have a strong beneficial deviation \(\sigma^\prime_2\) from this strategy. Now, see that:
\begin{itemize}
\item{If exactly \(\gamma_1\) and \(\gamma_2\) were satisfied in the resulting strategy --- \((\sigma^\prime_1, \sigma^\prime_2, \sigma''_3)\), then the joint strategy \((\sigma^\prime_1, \sigma^\prime_2)\) would be a punishment for \(\sigma^\prime_3\).}
\item{If exactly \(\gamma_2\) and \(\gamma_3\) were satisfied in \((\sigma^\prime_1, \sigma^\prime_2, \sigma''_3)\) then the joint strategy \((\sigma^\prime_2, \sigma''_3)\) would be a punishment for \(\sigma^\prime_1\).}
\end{itemize}
Hence, \(2\) must be the only winner in \((\sigma^\prime_1, \sigma^\prime_2, \sigma''_3)\).

Once again, this strategy must admit a strong beneficial deviation for \(1\), \(3\), or both. Let that deviation be \((\sigma''_1, \sigma'''_3)\) with possibly \(\sigma''_1\) = \(\sigma^\prime_1\) or \(\sigma'''_3\) = \(\sigma''_3\) but not both. Consider the following cases for the resulting strategy, \((\sigma''_1, \sigma^\prime_2, \sigma'''_3)\):
\begin{itemize}
\item{If only \(\gamma_1\), or exactly \(\gamma_1\) and \(\gamma_2\) were satisfied then we would necessarily have \(\sigma'''_3 = \sigma''_3\) and the joint strategy \((\sigma''_1, \sigma^\prime_2)\) would be a punishment for \(\sigma''_3\), contradiction.}
\item{If exactly \(\gamma_1\) and \(\gamma_3\) were satisfied then \((\sigma''_1, \sigma'''_3)\) would be a punishment for \(\sigma^\prime_2\), contradiction.}
\item{If only \(\gamma_3\), or exactly \(\gamma_2\) and \(\gamma_3\) were satisfied then we would necessarily have \(\sigma''_1 = \sigma^\prime_1\) and the joint strategy \((\sigma^\prime_2, \sigma'''_3)\) would be a punishment for \(\sigma^\prime_1\), contradiction.}
\end{itemize}
Thus all three possible sets of winners lead to a contradiction, yielding our result.
\end{proof}

The second of the two technical lemmas is then as follows:

\begin{lemma}\label{lmm:strong-core-three-player-non-empty-technical-lemma-2}
Let \(\{a,b,c\} = \{1,2,3\}\) be the set \(Ag\). There cannot be strategies \((\sigma_a, \sigma_b, \sigma_c)\), \((\sigma_a, \sigma_b, \sigma^\prime_c)\) of \(G\) such that all of the following hold:
\begin{itemize}
\item{\((\sigma_a, \sigma_b, \sigma_c)\) satisfies \(\gamma_a\) and \(\gamma_b\) but not \(\gamma_c\)}
\item{\((\sigma_a, \sigma_b, \sigma^\prime_c)\) satisfies \(\gamma_c\) only}
\item{\(\sigma^\prime_c\) is a strong beneficial deviation for \(c\) from \((\sigma_a, \sigma_b, \sigma_c)\) to \((\sigma_a, \sigma_b, \sigma^\prime_c)\)}
\end{itemize}
\end{lemma}

\begin{proof}
        For contradiction, assume otherwise. Without loss of generality, we can again assume that \(a=1\), \(b=2\), and \(c=3\). 

The second strategy, \((\sigma_a, \sigma_b, \sigma^\prime_c)\) must admit a strong beneficial deviation for \(1\), \(2\), or both. However, see that a deviation to a strategy satisfying only \(\gamma_1\) or only \(\gamma_2\) would contradict Lemma~\ref{lmm:strong-core-three-player-non-empty-technical-lemma-1}. Hence, the new strategy must have exactly two winners. Moreover, if the winners in the new strategy were exactly \(1\) and \(2\), the deviation would provide a punishment for \(\sigma^\prime_3\). Therefore, without loss of generality, we can assume that \((\sigma_a, \sigma_b, \sigma^\prime_c)\) has a strong beneficial deviation \(\sigma^\prime_1\) for \(1\) to a strategy satisfying exactly \(\gamma_1\) and \(\gamma_3\).

Now, the only loser in \((\sigma^\prime_1, \sigma_2, \sigma^\prime_3)\), \(2\), must have a strong beneficial deviation \(\sigma^\prime_2\). If however, exactly \(\gamma_1\) and \(\gamma_2\) were satisfied in \((\sigma^\prime_1, \sigma^\prime_2, \sigma^\prime_3)\), then \((\sigma^\prime_1, \sigma^\prime_2)\) would be a punishment for \(\sigma^\prime_3\). If exactly \(\gamma_2\) and \(\gamma_3\) were satisfied, \((\sigma^\prime_2, \sigma^\prime_3)\) would provide a punishment for \(\sigma^\prime_1\). Hence, \(2\) must be the only winner.

Once again a strong beneficial deviation from \((\sigma^\prime_1, \sigma^\prime_2, \sigma^\prime_3)\) to a strategy satisfying only \(\gamma_1\) or only \(\gamma_3\) would contradict Lemma~\ref{lmm:strong-core-three-player-non-empty-technical-lemma-1}. Therefore, there must be a strong beneficial deviation by \(1\), \(3\), or both, to a strategy with exactly two winners. Now:
\begin{itemize}
\item{If the winners were \(1\) and \(3\), the deviation would provide a punishment for \(\sigma^\prime_2\).}
\item{If the winners were \(1\) and \(2\) then the deviation must be \(\sigma''_1\) by \(1\) only (as \(3\) loses). But then, the joint strategy \((\sigma''_1, \sigma^\prime_2)\) would provide a punishment for the initial \(\sigma^\prime_3\).}
\item{If the winners were \(2\) and \(3\) then the deviation must be \(\sigma''_3\) by \(3\) only. The joint strategy \((\sigma^\prime_2, \sigma''_3)\) would be a punishment for \(\sigma^\prime_1\).}
\end{itemize}
This exhausts all the cases, each of which lead to a contradiction, and finishes the proof of Lemma~\ref{lmm:strong-core-three-player-non-empty-technical-lemma-2}.
\end{proof}

With these two technical lemmas in place, we are ready to prove the second part of our claim.

\begin{lemma}\label{lmm:strong-core-three-player-non-empty-second-part}
  Let \(G\) be a three-player game with an empty strong core. Then there does not exist a strategy profile \(\vec{\sigma}\) on \(G\) that models exactly two of \(\gamma_1\), \(\gamma_2\), \(\gamma_3\).
\end{lemma}

\begin{proof}
For a contradiction, assume otherwise. Without a loss of generality, let \((\sigma_1, \sigma_2, \sigma_3)\) be a strategy satisfying exactly \(\gamma_1\) and \(\gamma_2\). It must admit a strong beneficial deviation \(\sigma^\prime_3\) for the lone loser, \(3\). Moreover, by Lemma~\ref{lmm:strong-core-three-player-non-empty-technical-lemma-2}, he cannot win alone in the new strategy. Therefore, without loss of generality, we can assume that \(1\) and \(3\) are winners in \((\sigma_1, \sigma_2, \sigma^\prime_3)\).

Again, by Lemma~\ref{lmm:strong-core-three-player-non-empty-technical-lemma-2}, the new strategy must have a strong beneficial deviation \(\sigma^\prime_2\) by \(2\) to a strategy where exactly two players win. However, if \(1\) and \(2\) were winners, then \((\sigma_1, \sigma^\prime_2)\) would constitute a punishment for \(\sigma^\prime_3\). Therefore, exactly \(\gamma_2\) and \(\gamma_3\) must be satisfied in \((\sigma_1, \sigma^\prime_2, \sigma^\prime_3)\).

But now, again by Lemma~\ref{lmm:strong-core-three-player-non-empty-technical-lemma-2}, there must be a strong beneficial deviation \(\sigma^\prime_1\) for \(1\) to another strategy with exactly two winners. However, if those were \(1\) and \(2\), \((\sigma^\prime_1, \sigma^\prime_2)\) would be a punishment for \(\sigma^\prime_3\). Else, if the winners were \(1\) and \(3\) then \((\sigma^\prime_1, \sigma^\prime_3)\) would form a punishment for \(\sigma^\prime_2\). This concludes our proof.
\end{proof}

Finally, the required result simply follows by combining the two lemmas.

\begin{theorem}
  If \(G\) is a three-player game, then it has a non-empty strong core.
\end{theorem}

\begin{proof}
  Result follows by combining Lemma~\ref{lmm:strong-core-three-player-non-empty-first-part} and Lemma~\ref{lmm:strong-core-three-player-non-empty-second-part}.
\end{proof}

\end{document}